\documentclass[11pt]{article} 

\usepackage[margin=1in]{geometry} 
\usepackage[utf8]{inputenc}
\usepackage[T1]{fontenc}
\usepackage[sc]{mathpazo} 

\usepackage{authblk}
\usepackage{comment,cprotect,xspace}
\usepackage[normalem]{ulem} 

\usepackage{booktabs,multicol,multirow}
\usepackage[labelformat=simple]{subcaption}

\usepackage{enumitem} 
\usepackage[dvipsnames]{xcolor} 
\definecolor{ptblue}{RGB}{15,76,129} 
\definecolor{ptemerald}{HTML}{009473} 
\definecolor{ptgray}{HTML}{939597} 

\usepackage{tikz} 
\usetikzlibrary{arrows.meta}

\usepackage{amsmath,amsfonts,amssymb,amsthm,bbm,mathtools}

\let\OLDland\land
\renewcommand{\land}{\:\OLDland\:}
\let\OLDlor\lor
\renewcommand{\lor}{\:\OLDlor\:}
\let\OLDforall\forall
\renewcommand{\forall}{\;\OLDforall\:}
\let\OLDexists\exists
\renewcommand{\exists}{\;\OLDexists\,}

\newcommand*{\diff}[1]{\mathop{}\!\mathrm{d}#1} 
\DeclareMathOperator*{\argmax}{arg\,max}

\usepackage[ruled,linesnumbered,vlined]{algorithm2e}
\let\oldnl\nl
\newcommand{\nonl}{\renewcommand{\nl}{\let\nl\oldnl}}

\usepackage[square,sort]{natbib} 

\usepackage{hyperref}
\hypersetup{
linktocpage,
colorlinks=true,
citecolor=ptemerald, 
urlcolor=ptblue, 
linkcolor=Plum, 
}
\usepackage{cleveref} 

\theoremstyle{plain}
\newtheorem{theorem}{Theorem}[section]
\newtheorem{claim}[theorem]{Claim}

\newtheorem{lemma}[theorem]{Lemma}
\newtheorem{proposition}[theorem]{Proposition}

\theoremstyle{definition}
\newtheorem{definition}[theorem]{Definition}

\theoremstyle{remark}
\newtheorem*{remark}{\upshape\bfseries Remark}

\definecolor{cobalt}{rgb}{0.0, 0.28, 0.67}

\newcommand{\A}{\mathcal{A}}

\makeatletter 
\newcommand{\EF}[1]{\if\relax\detokenize\expandafter{\@firstofone#1{}}\relax \text{EF}\xspace\else \text{EF#1}\fi}
\makeatother
\newcommand{\EFone}{\EF{1}\xspace}
\newcommand{\EFX}{\EF{X}\xspace}
\newcommand{\EFM}{\EF{M}\xspace}

\newcommand{\EFable}{\text{envy-freeable}\xspace}
\newcommand{\EFability}{envy-freeability\xspace}

\newcommand{\goodMinimal}{good-minimal\xspace}
\newcommand{\goodMinimality}{good-minimality\xspace}
\newcommand{\choreMaximal}{chore-maximal\xspace}
\newcommand{\choreMaximality}{chore-maximality\xspace}

\newcommand{\Inj}{\mathrm{Inj}}

\newcommand{\val}{\mathrm{val}}
\newcommand{\OPT}{\operatorname{OPT}}

\newcommand{\IMWM}{\texttt{\textup{IMWM}}\xspace} 
\newcommand{\IMWPM}{\texttt{\textup{IMWPM}}\xspace}

\title{Fair Division with Indivisible Goods, Chores, and Cake}

\author[1]{Haris Aziz}
\author[2]{Xinhang Lu}
\author[1]{Simon Mackenzie}
\author[1]{Mashbat Suzuki}

\affil[1]{UNSW Sydney, Australia, \nolinkurl{haris.aziz@unsw.edu.au}, \nolinkurl{simon.william.mackenzie@gmail.com}, \nolinkurl{mashbat.suzuki@unsw.edu.au}}
\affil[2]{Kyushu University, Japan, \nolinkurl{xinhang.lu@inf.kyushu-u.ac.jp}}

\date{}

\begin{document}
\maketitle

\begin{abstract}
We study the problem of fairly allocating indivisible items and a desirable heterogeneous divisible good (i.e., cake) to agents with additive utilities.
In our paper, each indivisible item can be a good that yields non-negative utilities to some agents and a chore that yields negative utilities to the other agents.
Given a fixed set of divisible and indivisible resources, we investigate almost envy-free allocations, captured by the natural fairness concept of \emph{envy-freeness for mixed resources (EFM)}.
It requires that an agent~$i$ not envy another agent~$j$ if agent~$j$'s bundle contains any piece of cake yielding positive utility to agent~$i$ (i.e., envy-freeness), and that agent~$i$ be envy-free up to one item (EF1) towards agent~$j$ otherwise.
We prove that with indivisible items and a cake, an EFM allocation always exists for any number of agents with additive utilities.

We reduce the problem of showing the existence of EFM to the problem of proving the compatibility of EF1 and envy-freeability.
An allocation of indivisible items is said to be \emph{envy-freeable} if the allocation can be made envy-free by allocating a sufficient amount of money to each of the agents.
We show that for indivisible items, an EF1 and envy-freeable allocation always exists.
Moreover, such an allocation can be found in polynomial time.
This result, in and of itself, is of independent interest.
\end{abstract}

\section{Introduction}

Fair division studies the fundamental problem of allocating valuable resources or undesirable burdens among agents with possibly differing preferences in a fair manner.
The problem was first formulated mathematically in the seminal work of \citet{Steinhaus49}, which addresses the fair division of a heterogeneous divisible resource, commonly known as \emph{cake cutting}.
Since then, fair division has attracted ongoing interest in mathematics, political science, and economics~\citep{BramsTa96}, and most recently has received extensive attention in computer science.

A thorough understanding of fairness has been described as a key to ``addressing some of the greatest challenges of the 21st century''~\citep{Procaccia13}.
Among the plethora of fairness criteria proposed in the literature, \emph{envy-freeness} has been regarded as the ``gold standard'' of fairness in allocation problems~\citep[e.g.,][]{BenadeKaPr24}.
It requires that each agent weakly prefer their own bundle to any other agent's bundle.
An envy-free cake division always exists~\citep{Alon87,Su99}, and can be found via a discrete and bounded protocol~\citep{AzizMa16-STOC,AzizMa16} under the query model of \citet{RobertsonWe98}.

Looking beyond cake cutting, the allocation of heterogeneous \emph{indivisible} items is of practical importance and has attracted considerable attention over the past decades, especially in the context of allocating goods (items with non-negative utilities)~\citep{AmanatidisAzBi23,NguyenRo23,Suksompong21,Suksompong25}.
More recently, due to its broad applicability, the fair allocation of chores (items that yield negative utilities) has also received significant attention~\citep{GuoLiDe23}.

There has been a concerted effort to understand whether strong and compelling fairness guarantees can be achieved in the more general setting that includes a mix of both goods and chores~\citep[e.g.,][]{AzizBuLu27,AzizCaIg22,AMS20a,BarmanHVSe25,BarmanVe26,BogomolnaiaMoSa17,ChaudhuryGaMc23,HosseiniSiVa23,KMT21a,KMT21b,LiuLuSu24}.
In this setting, an item can be a good that yields positive utilities to some agents and a chore that yields negative utilities to others.
In order to circumvent the issue that an envy-free allocation may not exist with indivisible items, envy-freeness is often relaxed to \emph{envy-freeness up to one item (\EFone)}~\citep{AzizCaIg22,Budish11,LiptonMaMo04}, a highly influential, robust, and widely studied concept.
\EFone requires that any envy an agent has towards another agent vanishes if we ignore some chore in the former agent's bundle or some good in the latter agent's bundle.
An \EFone allocation of mixed indivisible goods and chores always exists and can be computed efficiently for agents with additive or even doubly monotone utilities~\citep{AzizCaIg22,BhaskarSrVa21}.

What if we are presented with both indivisible items and a cake?
This model of mixed divisible and indivisible resources captures real-world applications such as rent division, divorce settlement, inheritance division, and any other situation involving both assets and liabilities to be divided between multiple entities.
A compelling relaxation of envy-freeness in the mixed-resources model is \emph{envy-freeness for mixed resources (\EFM)}~\citep{BeiLiLi21,BhaskarSrVa21}, which simultaneously combines the ideas of envy-freeness and \EFone in a meaningful and natural way.
Phrased in our setting, \EFM requires that if an agent~$i$ envies another agent~$j$, then agent~$j$ should not get any sliver of the divisible resource for which agent~$i$ has positive utility, and moreover, agent~$i$ should be \EFone towards agent~$j$.
When all items are divisible, \EFM coincides with envy-freeness; when the items are all indivisible, \EFM coincides with \EFone.
With indivisible goods and a cake, an \EFM allocation always exists for any number of agents with additive utilities~\citep{BeiLiLi21}.

\citet{BhaskarSrVa21} revisited the work of \citet{BeiLiLi21} and examined the existence of \EFM allocations in the presence of indivisible chores and a cake.
\citeauthor{BhaskarSrVa21} were only able to establish the existence of \EFM allocations in two special cases: when agents have identical rankings over indivisible chores, and when the number of chores exceeds the number of agents by at most one.
\citet[p.~15]{BhaskarSrVa21} discussed in detail the technical challenges concerning the problem and emphasized that the solution will require fundamentally novel techniques.
The problem has again been highlighted in the survey paper by \citet[Open Question~7]{LiuLuSu24}.
The fundamental question that we explore in this paper and which has been raised in previous work~\citep{BhaskarSrVa21,LiuLuSu24,MPS23a} is as follows.
\begin{quote}
\itshape
With indivisible items and a cake, does an \EFM allocation always exist?
\end{quote}

\subsection{Our Contribution}

We settle the above question in the affirmative by establishing the existence of \EFM allocations in full generality for the case where there is a cake and $m$ indivisible items (that may be subjective goods or chores) to be divided among $n$ agents who have additive utilities.

\medskip
\noindent\textbf{Main Result~1 (\Cref{thm:EFM-existence}).}
For indivisible items and cake, an \EFM allocation always exists.
\medskip

En route, we uncover an interesting connection among the solution concepts of envy-freeness, \emph{\EFability}, \EFone, and \EFM, where \EFability is a concept developed and investigated in the closely related but separate line of research on fair division with subsidy~\citep[e.g.][]{BrustleDiNa20,HalpernSh19,LuMaSu26}.
An allocation of indivisible items is said to be \emph{\EFable} if the allocation can be made envy-free by allocating a sufficient amount of money to each of the agents~\citep{HalpernSh19}.
With only indivisible goods, \citet{BrustleDiNa20} showed that an \EFone and \EFable allocation always exists.

We reduce the problem of showing the existence of an \EFM allocation given any set of indivisible items and any cake to the problem of proving the existence of an \EFone and \EFable allocation of the indivisible items.
We establish the following result, extending the result of \citet{BrustleDiNa20} to the setting that includes both goods and chores.

\medskip
\noindent\textbf{Main Result~2 (\Cref{thm:EF1+EFable-existence}).}
For indivisible goods and chores, an \EFone and \EFable allocation always exists.
Moreover, such an allocation can be computed in polynomial time.
\medskip

Our results significantly and simultaneously generalize several fundamental results in the literature, including the following:
\begin{itemize}
\item For cake, an envy-free allocation is guaranteed to exist~\citep{Alon87,Su99}.
\item For indivisible goods, an \EFone allocation exists~\citep{LiptonMaMo04,CaragiannisKuMo19}.
\item For indivisible goods and chores, an \EFone allocation exists~\citep{AzizCaIg22,BhaskarSrVa21}.
\item For indivisible goods, an \EFone and \EFable allocation exists and can be computed in polynomial time~\citep{BrustleDiNa20}.
\item For indivisible goods and cake, an \EFM allocation exists~\citep{BeiLiLi21}.
\end{itemize}

\subsection{Technical Challenges and Overview}

As \citet{BhaskarSrVa21} pointed out, existing methods for obtaining \EFM allocations fail in the setting that includes a mix of indivisible goods and chores due to the necessity of resolving \emph{any} envy cycle, which may violate the \EFone property needed for \EFM.
We circumvent this issue by finding a fixed allocation of the indivisible items that will not be altered when we allocate the divisible good, while still achieving \EFM.
To achieve this, we first observe that the problem of allocating a heterogeneous cake and indivisible items is reducible to the problem of allocating indivisible items and money by showing that one can simulate money via cake.
Hence, to achieve \EFM, it suffices to find an indivisible-items allocation that is \EFone and can be made envy-free when there is a sufficient amount of money (i.e., an \EFable allocation).
We establish this in \Cref{sec:cake-to-money}.

Having reduced the problem of showing the existence of \EFM to finding a fixed allocation of the indivisible items that is both \EFone and \EFable, we face the challenge of identifying such allocations, which presents its own technical difficulties.
When there are only indivisible goods, \citet{BrustleDiNa20} showed that such allocations exist.
However, their algorithm violates \EFone in the setting with both goods and chores.
On the other hand, all existing methods for obtaining \EFone allocations for goods and chores (such as those presented in \citep{AzizCaIg22,BhaskarSrVa21}) violate \EFability.
We overcome this by carefully bundling items together, treating them as \emph{meta-items}, and then running an iterative matching algorithm similar to that of \citet{BrustleDiNa20} in a way that yields both \EFone and \EFability.

The key technical contribution of our paper is to design polynomial-time algorithms, in the setting of both goods and chores, that return an \EFone and \EFable allocation by carefully bundling items together.
We achieve this by starting with an initial bundling of the items (\Cref{alg:itemmerge}), which outputs a set of meta-items, each of which is considered good by some agent, along with the remaining individual objective chores.
We then branch into two cases depending on the number of remaining individual chores: (i) the number of objective chores is at least~$n$ (see \Cref{sec:EFM:>=n-chores}), and (ii) the number of objective chores is less than~$n$ (see \Cref{sec:EFM:<n-chores}).

In the first case, each meta-good can be combined with a distinct objective chore such that the resulting instance consists solely of objective chores.
We construct the pairing and the allocation together, maximizing each round's matching value subject to the matchings already fixed and the feasibility of completing the pairing; a maximum-profit flow implements each round in polynomial time.
We prove that this ensures \EFone and \EFability (\Cref{prop:EF1+EFable:>=n-chores}).

In the second case, we further refine the bundling (\Cref{alg:chore-max-good-min}), which may partially undo the initial bundling, and consider different possible configurations of combining the resulting meta-items.
This case fundamentally differs from the first one, as some agents may receive a meta-good.
We prove that there exists a bundling that cleanly separates the agents, ensuring that if an agent receives a meta-good, then every other agent who receives a meta-chore perceives that agent’s bundle as a meta-chore as well.
To construct this bundling when at least one objective chore remains, we use a greedy expansion procedure (\Cref{alg:chore-expansion}) that jointly chooses a matching and one meta-chore to expand.
Each iteration is implemented by a maximum-profit flow, and there are at most $n$ expansions.
The procedure returns the constructed meta-chores together with a maximum-weight matching that assigns them to distinct agents.
An alternating-path argument ensures that each recipient values its attachment set non-negatively, even when later iterations reassign the meta-chores.
We then allocate the remaining meta-goods and subjective goods to the remaining agents by iterative maximum-weight matching.
Together, these steps yield an \EFone and \EFable allocation (\Cref{prop:EF1+EFable:<n-chores:>0}).

\subsection{Related Work}

The fair division literature is vast and can be divided along at least two orthogonal directions: (i) goods versus chores and (ii) indivisible versus divisible items.
We refer interested readers to the survey articles of \citet{AmanatidisAzBi23}, \citet{GuoLiDe23}, \citet{NguyenRo23}, and \citet{Suksompong21,Suksompong25} for an overview of the recent developments in the fair allocation of indivisible goods or indivisible chores.
For settings with mixed indivisible goods and chores, or with mixed divisible and indivisible resources, we refer to the recent survey article of \citet{LiuLuSu24} for an overview.

Our paper is mostly inspired by the work of \citet{BeiLiLi21}, \citet{BhaskarSrVa21}, and \citet{BrustleDiNa20}.
We will discuss in detail the recent developments in related lines of research.

\paragraph{Fair Division of Mixed Resources.}
This line of research focuses on the axiomatic study of fairness when allocating a mix of divisible and indivisible resources (i.e., \emph{mixed resources}).
\citet{BeiLiLi21} were the first to study an envy-freeness relaxation in the setting with only goods and introduced a version of \EFM that is artificially stronger than our \EFM definition.\footnote{\citet[Definition~2.3]{BeiLiLi21} requires envy-freeness from an agent~$i$ to another agent~$j$ as long as agent~$j$'s bundle contains a positive amount of divisible goods, even in the case that agent~$i$ values the divisible goods in~$j$'s bundle at~$0$.
It is also known that the definition of \citeauthor{BeiLiLi21} is incompatible with Pareto optimality, a fundamental economic efficiency notion in the field.
Our \EFM definition is arguably more natural.}
Subsequent work has delved further into understanding \EFM by investigating the design and limitations of \emph{truthful} and \EFM mechanisms in the presence of strategic agents~\citep{LiLiLu23}, quantitatively measuring the social welfare loss resulting from enforcing \EFM through the lens of \emph{price of fairness}~\citep{LiLiLu24}, and exploring how to simultaneously satisfy exact fairness (e.g., envy-freeness and proportionality) ex ante while ensuring \EFM ex post~\citep{BuLiLi24}.
Other fairness notions, including the \emph{maximin share (MMS) guarantee}~\citep{BeiLiLu21}, relaxations of envy-freeness and proportionality based on \emph{indivisibility ratio}~\citep{LiLiLi24}, and a strengthening of \EFM to EFXM~\citep{NishimuraSu25}, have also been studied in the mixed-goods model.

\citet{BhaskarSrVa21} generalized the mixed-goods model and the \EFM notion further to the mixed-resources model in which items could be goods or chores.
In the setting with \emph{doubly monotone} indivisible items and a divisible chore,\footnote{\label{ft:doubly-monotone}With doubly monotone indivisible items, each agent can partition the set of items into objective goods and objective chores.} \citeauthor{BhaskarSrVa21} showed an \EFM allocation always exists.

Slightly further afield, \citet{AzizHeLu25} and \citet{BeiLiLu25} relaxed the assumption that agents agree on whether each resource is divisible or indivisible.

\paragraph{Fair Division with Subsidy.}
The study of fair allocation with monetary transfers is well-established, particularly in the context of room-rent division.
Most of the literature in this area assumes that each agent demands exactly one good or room~\citep{Arag95a,GMPZ17a,Klij00a,Sven83a}.
More general frameworks, where envy-freeness is achieved through side-payments, have been explored~\citep{HRS02a}.
More recently, research has focused on computing envy-free allocations when agents may demand multiple goods and monetary transfers are permitted.
In particular, \citet{HalpernSh19} brought attention to the problem of determining allocations that achieve envy-freeness through minimal subsidies.
Subsequent work has delved deeper into the computational aspects of this problem, examining both exactly minimal subsidies~\citep{BrustleDiNa20,KawaseMaSu25} and approximately minimal subsidies~\citep{CaragiannisIo21}, as well as its strategic aspects, examining fair and truthful mechanisms with limited subsidy~\citep{GokoIgKa24}.

While those works consider a goods-only setting, recently, \citet{LuMaSu26} studied envy-freeness with subsidies in the setting with mixed goods and chores.
They showed that when the value of each item is bounded between~$-1$ and~$1$, a subsidy of at most one dollar per agent suffices to guarantee an envy-free allocation.
It is worth noting that the compatibility between \EFone and \EFability (i.e., our Main Result~2) does not necessarily lead to the desired subsidy bound.
At the same time, the indivisible-items allocation obtained by \citet{LuMaSu26} to achieve the subsidy bound may not satisfy \EFone.

\section{Preliminaries}
\label{sec:prelim}

We study fair allocation of indivisible items (consisting of both goods and chores) and a divisible good among a set of agents $N = \{1, \dots, n\}$.
For any positive integer~$k$, let $[k] \coloneqq \{1, 2, \dots, k\}$.

\medskip
\noindent\textbf{Indivisible Items.}
Denote by $M = [m]$ the set of \emph{indivisible items}.
Each agent~$i \in N$ has a utility function $u_i \colon 2^M \to \mathbb{R}$ that assigns a real value to each subset of items, with  $u_i(\emptyset) = 0$ for all~$i \in N$.
Agents' utilities are assumed to be \emph{additive}, i.e., $u_i(S) = \sum_{t \in S} u_i(t)$ for each bundle~$S \subseteq M$.
For convenience, we write $u_i(t)$ to denote $u_i(\{t\})$ for a single item~$t \in M$.

An item~$t \in M$ can be a \emph{good} for one agent~$i$ (i.e., $u_i(t) \geq 0$) but a \emph{chore} for another agent~$j$ (i.e., $u_j(t) < 0$).
An item~$t$ is said to be a \emph{subjective good} if $u_i(t) \geq 0$ for some~$i \in N$.
An item~$t$ is said to be an \emph{objective good} (resp., \emph{objective chore}) if $u_i(t) \geq 0$ (resp., $u_i(t) < 0$) for all agents~$i \in N$.

We refer to a setting with agents~$N$, indivisible items~$M$, and utilities~$(u_i)_{i \in N}$ as an \emph{indivisible-items instance} (henceforth, \emph{instance}), denoted as $\langle N, M, (u_i)_{i \in N} \rangle$.

\medskip
\noindent\textbf{Divisible Resource.}
We consider two types of divisible resources: cake (a heterogeneous divisible resource) and money (a homogeneous divisible resource).
We now formally define each of them.

A \emph{cake}~$C$ is a heterogeneous divisible good represented by the unit interval $C = [0, 1]$.
Each agent~$i \in N$ has a non-atomic \emph{density function} $f_i \colon [0, 1] \to \mathbb{R}_{\geq 0}$ that describes their preferences over the cake.
For any measurable subset~$X \subseteq [0, 1]$, $i$'s utility for~$X$ is given by $u_i(X) = \int_{x \in X} f_i(x) \, \diff x$.

\emph{Money} is a homogeneous divisible resource.
If agent~$i \in N$ receives $p_i \geq 0$ units of money, their utility for this money is $u_i(p_i) = p_i$.
That is, all agents value money identically at one unit of utility per unit of money.
Note that money can be seen as a special case of cake, where each agent's density function is identically one.

\medskip
\noindent\textbf{Precision and Input Representation.}
When discussing the computational aspects, we assume that $u_i(t)$ for each~$i \in N, t \in M$ and $u_i(C)$ for each~$i \in N$ are all rational numbers.

\medskip
\noindent\textbf{Allocations.}
An \emph{indivisible-items allocation} $A = (A_1, \dots, A_n)$ is a partition of~$M$ such that $A_i \cap A_j = \emptyset$ for all $i \neq j$ and $\bigcup_{i \in N} A_i = M$.
Each agent~$i$ receives bundle~$A_i$, which could be empty.

An allocation of mixed divisible and indivisible resources is a tuple $\mathcal{A} = \{(A_i, C_i)\}_{i \in N}$ where $(A_1, \dots, A_n)$ is a partition of~$M$ and $(C_1, \dots, C_n)$ is a partition of~$[0, 1]$ into measurable sets.
Agent~$i$ receives indivisible items~$A_i$ and a piece of cake~$C_i$, with total utility $u_i(A_i, C_i) \coloneqq u_i(A_i) + u_i(C_i)$.

When the divisible resource is money, we denote an allocation of indivisible items and money as $\mathcal{A} = \{(A_i, p_i)\}_{i \in N}$, where agent~$i$ receives indivisible items~$A_i$ and $p_i \geq 0$ units of money.

\subsection{Fairness Notions}

In this work, we are concerned with the problem of fairly dividing a set of indivisible items and a cake.
We begin with classical fairness notions before introducing our main concepts.

\begin{definition}[Envy-freeness]
With a mix of indivisible items~$M$ and a cake~$C$, an allocation $\{(A_i, C_i)\}_{i \in N}$ is said to be \emph{envy-free} if for every pair of agents~$i, j \in N$, $u_i(A_i, C_i) \geq u_i(A_j, C_j)$.
\end{definition}

When the set of indivisible items is empty and we only divide cake~$C$, the problem is known as \emph{cake cutting}.
An envy-free cake cutting always exists~\citep{Alon87,Su99}.
With only indivisible items, however, envy-freeness is a demanding property and cannot always be satisfied.
It is thus often relaxed to \emph{envy-freeness up to one item (\EFone)}.
An \EFone allocation of indivisible items always exists~\citep{AzizCaIg22,BhaskarSrVa21}.

\begin{definition}[\EFone]
An indivisible-items allocation~$(A_i)_{i \in N}$ is said to be \emph{envy-free up to one item (\EFone)} if for all~$i, j \in N$,
\begin{itemize}
\item either $u_i(A_i) \geq u_i(A_j)$, or
\item there exists an item $t \in A_i \cup A_j$ such that $u_i(A_i \setminus \{t\}) \geq u_i(A_j \setminus \{t\})$.
\end{itemize}
\end{definition}

\paragraph{Fair Division of Mixed Divisible and Indivisible Resources.}
We are now ready to introduce our central fairness concept, which works for the setting with both divisible and indivisible resources and combines envy-freeness and \EFone naturally.

\begin{definition}[\EFM]
\label{def:EFM}
With indivisible items~$M$ and cake~$C$, an allocation $\{(A_i, C_i)\}_{i \in N}$ is said to be \emph{envy-free for mixed resources (\EFM)} if for every pair of agents~$i, j \in N$, either agent~$i$ is envy-free toward agent~$j$, i.e., $u_i(A_i) + u_i(C_i) \geq u_i(A_j) + u_i(C_j)$, or both of the following conditions hold:
\begin{itemize}
\item Agent~$i$ does not value the cake allocated to agent~$j$, i.e., $u_i(C_j) = 0$; and
\item There exists an item $t \in A_i \cup A_j$ such that $u_i(A_i \setminus \{t\})  \geq u_i(A_j \setminus \{t\}) $.
\end{itemize}
\end{definition}

\begin{remark}
In the second bullet point of \Cref{def:EFM}, we could also include agent~$i$'s utility over her allocated divisible good in the definition; ``weak \EFM'' of \citet[Definition~6.4]{BeiLiLi21} is defined in this way.
However, our \Cref{def:EFM} is slightly stronger and can always be satisfied.
\end{remark}

Since \EFM is a comparison-based fairness notion, it is scale-invariant: multiplying an agent's utility function by any positive constant uniformly across both divisible and indivisible items does not affect \EFM.
Hence, without loss of generality, by applying such a uniform scaling, we may assume that for all agents, either $u_i(C) = 1$ or $u_i(C) = 0$.
Throughout the paper, we adopt this normalization when studying the existence of \EFM allocations.

When the \EFM definition is applied to a homogeneous cake (money), an allocation of indivisible items and money, denoted as $\mathcal{A} = \{(A_i, p_i)\}_{i \in N}$, is said to be \EFM if, for every pair of agents~$i, j \in N$, either agent~$i$ is envy-free toward agent~$j$, i.e., $u_i(A_i) + p_i \geq u_i(A_j) + p_j$, or both of the following conditions hold:
\begin{itemize}
\item agent~$j$ does not receive any money, i.e., $p_j = 0$; and
\item there exists an item $t \in A_i \cup A_j$ such that $u_i(A_i \setminus \{t\}) \geq u_i(A_j \setminus \{t\})$.
\end{itemize}

One might wonder if we can further strengthen \EFM by replacing the \EFone criterion with the stronger \EFX criterion; see the corresponding definition below.

\begin{definition}[EFXM]
\label{def:EFXM}
With indivisible items~$M$ and cake~$C$, an allocation $\{(A_i, C_i)\}_{i \in N}$ is said to be \emph{envy-free up to any item for mixed resources (EFXM)} if for every pair of agents~$i, j \in N$, either agent~$i$ is envy-free toward agent~$j$, i.e., $u_i(A_i) + u_i(C_i) \geq u_i(A_j) + u_i(C_j)$, or all of the following conditions hold:
\begin{itemize}
\item There is no $C' \subseteq C_j$ such that $u_i(C') > 0$; and

\item For each $t \in \{g \in A_j \mid u_i(g) > 0\} \cup \{c \in A_i \mid u_i(c) < 0\}$, it is the case that $u_i(A_i \setminus \{t\}) \geq u_i(A_j \setminus \{t\})$.
\end{itemize}
\end{definition}

It would be interesting to explore to what extent an allocation satisfying EFXM is guaranteed to exist.
For instance, with both divisible and indivisible goods, \citet{NishimuraSu25} showed that for binary linear utilities, a maximum Nash welfare allocation satisfies EFXM.

In our setting with both goods and chores, however, \citet{HeTa27} showed that an \EFX allocation may not exist, even with only objective chores and under tri-valued additive utilities.
When there is no cake, EFXM coincides with EFX; hence, by the result of \citeauthor{HeTa27}, we see that EFXM is not guaranteed to exist.
In this sense, \EFM is the strongest known concept for which a general existence result holds for agents with additive utilities in the presence of indivisible goods and chores.

\paragraph{Fair Division of Indivisible Items with Subsidy.}
\citet{HalpernSh19} formalized the concept of \emph{\EFability} to capture allocations of indivisible goods that can be made envy-free through monetary compensation from a third-party.
While their definition and characterization of \EFability were framed in the setting with only indivisible goods, they can be applied straightforwardly to the setting with mixed indivisible goods and chores.

\begin{definition}[Envy-freeability]
An indivisible-items allocation~$A = (A_1, \dots, A_n)$ is said to be \emph{envy-freeable} if there exists a vector of payments~$(p_1, \dots, p_n)$ such that $\{(A_i, p_i)\}_{i \in N}$ is envy-free.
\end{definition}

\citet{HalpernSh19} gave a useful characterization of \EFability using the envy graph.
Given an indivisible-items allocation $A = (A_i)_{i \in N}$ of~$M$, its \emph{envy graph}~$\mathcal{G}_A$ is the complete weighted directed graph in which each agent is a node, and each edge~$(i, j)$ has weight $w_A(i, j) \coloneqq u_i(A_j)- u_i(A_i)$, i.e., the amount of envy that agent~$i$ has for agent~$j$.
The weight of any (directed) path~$P$ is defined as $w_A(P) = \sum_{(i, j) \in P} w_A(i, j)$, i.e., the sum of the weights of the edges along the path.

\begin{theorem}[\citealp{HalpernSh19}]
\label{thm:hs_ef}
For any allocation $A = (A_1, \dots, A_n)$ of indivisible items, the following statements are equivalent:
\begin{enumerate}[label=(\roman*)]
\item Allocation~$A$ is envy-freeable.
\item Allocation~$A$ maximizes the social welfare across all reassignments of its bundles among the agents.
That is, for every permutation~$\sigma$ over~$N$, we have $\sum_{i = 1}^n u_i(A_i) \geq \sum_{i = 1}^n u_i(A_{\sigma(i)})$.
\item The envy graph~$\mathcal{G}_A$ has no positive-weight cycles.
\end{enumerate}
\end{theorem}

Note that an \EFable allocation need not be fair (e.g., \EFone).
It only requires the indivisible-items allocation combined with enough subsidy payments to be envy-free.
Notably, with only indivisible goods, \citet{BrustleDiNa20} show that an \EFone and \EFable allocation always exists.
We establish below an analogous result in the setting with only indivisible chores.

\begin{proposition}
\label{prop:ind-chores:EF1+EFable}
With only indivisible chores, an \EFone and \EFable allocation always exists.
Moreover, such an allocation can be found in polynomial time.
\end{proposition}

At a high level, this result follows from first padding the instance with dummy chores so that the number of items is a multiple of~$n$, and then running an \emph{iterative maximum-weight perfect matching} on an appropriately defined graph.
This graph has the property that each partial allocation obtained in each round is \EFable.
Due to the additivity of the utilities, when adding up the partial allocations across all rounds, \EFability is preserved.
By the nature of maximum-weight perfect matching, each agent prefers the item they are given in the current round to any item available in subsequent rounds.
Finally, \EFone follows since each agent may remove their chore from the last round to achieve envy-freeness towards any other agent.

\begin{proof}[Proof of \Cref{prop:ind-chores:EF1+EFable}]
\begin{algorithm}[t]
\caption{Iterative Perfect Matching for Indivisible Chores}
\label{alg:IMWM-chores}
\DontPrintSemicolon

\KwIn{Agents~$N$, indivisible chores~$M$ with $|M| = T \cdot n - k$ and agents' utilities $(u_i)_{i \in N}$.}
\KwOut{An \EFone and \EFable allocation~$A = (A_1, \dots, A_n)$.}

$A_i \gets \emptyset$ for all~$i \in N$.\;
Extend~$M$ to~$\widehat{M}$ by adding $k$ dummy items, each of which is valued at~$0$ for all agents.\;
$t \gets 1$; $J_1 \gets \widehat{M}$\;
\While{$J_t \neq \emptyset$}{
    Compute a maximum-weight perfect matching $\mu^t = \{(i, \mu_i^t)\}$ in $H[N, J_t]$.\;
    $A_i \gets A_i \cup \{\mu_i^t\}$ for all~$i \in N$.\;
    $J_{t+1} \gets J_t \setminus \bigcup_{i \in N} \{\mu_i^t\}$\;
    $t \gets t + 1$\;
}
\Return{allocation~$A = (A_1, A_2, \dots, A_n)$}
\end{algorithm}

We will show that \Cref{alg:IMWM-chores} outputs an allocation that is both \EFone and \EFable.
Since there are exactly $T \cdot n$ chores in~$\widehat{M}$, \Cref{alg:IMWM-chores} terminates in~$T$ rounds.
Let $J = \widehat{M}$, and the graph $H[N, J]$ be a complete weighted bipartite graph.
For each~$i \in N$ and~$c \in J$, edge~$(i, c)$ has weight~$u_i(c)$.
In each round~$t$, each agent~$i \in N$ is matched to a chore~$\mu_i^t$ (which could be a dummy item).
As a result, each agent~$i \in N$ receives a bundle $A_i = \{\mu_i^1, \dots, \mu_i^T\}$.

\medskip
\noindent\textbf{Allocation~$A$ is \EFone.}
Observe that $u_i(\mu_i^t) \geq u_i(c)$ for any chore~$c \in J_{t+1}$; otherwise, replacing the edge~$(i, \mu_i^t)$ with $(i, c)$ in~$\mu^t$ results in a matching of strictly larger weight in $H[N, J_t]$.
For any pair of agents~$i, j \in N$, we have
\begin{align*}
u_i(A_i \setminus \{\mu_i^T\}) = u_i(\{\mu_i^1, \dots, \mu_i^{T - 1}\})
&= u_i(\mu_i^1) + \cdots + u_i(\mu_i^{T-1}) \\
&\geq u_i(\mu_j^2) + \cdots + u_i(\mu_j^T) \\
&\geq u_i(\mu_j^1) + u_i(\mu_j^2) + \cdots + u_i(\mu_j^T) \\
&= u_i(A_j),
\end{align*}
meaning that the output allocation~$A$ of \Cref{alg:IMWM-chores} is \EFone.

\medskip
\noindent\textbf{Allocation~$A$ is \EFable.}
The argument follows exactly the proof of Lemma~3.1 in the paper by \citet{BrustleDiNa20}.
For the sake of being self-contained, we still provide the proof below.
By \Cref{thm:hs_ef}, it suffices to show envy graph~$\mathcal{G}_A$ has no positive-weight cycles.
Take any directed cycle~$\gamma$ in the envy graph~$\mathcal{G}_A$ and assume without loss of generality that $\gamma$ involves a sequence of nodes/agents $1, 2, \dots, k$ for some $k \geq 2$.
We have
\begin{align*}
w_{A}(\gamma) = \sum_{(i, j) \in \gamma} w_{A}(i, j) = \sum_{(i, j) \in \gamma} \left( u_i(A_j) - u_i(A_i) \right)
&= \sum_{(i, j) \in \gamma} \sum_{t = 1}^T \left( u_i(\mu^t_j) - u_i(\mu^t_i) \right) \\
&= \sum_{(i, j) \in \gamma} \sum_{t = 1}^T w_{\mu^t}(i, j) = \sum_{t = 1}^T \sum_{(i, j) \in \gamma} w_{\mu^t}(i, j)
\end{align*}

Let $\pi_\gamma$ be the permutation of~$N$ under which $\pi_\gamma(i) = i + 1$ for each $i \in [k-1]$, $\pi_\gamma(k) = 1$, and $\pi_\gamma(i) = i$ for $i \in N \setminus [k]$.
In each round~$t$, $\mu^t$ is a maximum-weight matching; therefore, $\sum_{(i, j) \in \gamma} w_{\mu^t}(i, j) \leq 0$.
Otherwise, giving each~$i$ the item $\mu^t_{\pi_\gamma(i)}$ results in a matching of strictly larger weight, a contradiction.
Thus, $w_A(\gamma) \leq 0$, implying that allocation~$A$ is \EFable.

\medskip
\noindent\textbf{Polynomial-Time Computation.}
Let $m=|M|$. If $m=0$, the empty allocation suffices. Otherwise, choose $T=\lceil m/n\rceil$ and $k=Tn-m$, so that $0\leq k<n$ and the padded instance contains $Tn<m+n$ items. There are exactly $T\leq m$ matching rounds. In each round, a maximum-weight matching covering all $n$ agents can be computed in polynomial time on a complete bipartite graph with at most $m+2n$ vertices and at most $n(m+n)$ edges. For rational input utilities encoded in binary, all edge weights are either input utilities or zero, and hence have polynomial bit length. Using exact rational arithmetic, the matching computations and the allocation updates therefore take time polynomial in the input bit length.
\end{proof}

It is obvious that the goal of the subsidy literature is different from the axiomatic study of fairness for fair division with mixed divisible and indivisible resources.
Nevertheless, in the next section, we will show that the tools developed in the subsidy literature are very useful to prove the existence of \EFM allocations when we divide a fixed set of divisible and indivisible resources.

\section{Cake to Money Reduction}
\label{sec:cake-to-money}

In this section, we show that, when considering the existence of \EFM allocations, the problem of allocating a cake among agents with heterogeneous valuations can be reduced to the problem of allocating money.
A defining characteristic of money is that all agents value it equally; a dollar holds the same worth regardless of who possesses it.
In contrast, when agents value the cake heterogeneously, the values assigned to different portions of the cake may vary arbitrarily across agents.
Despite these differences, the following result demonstrates that, with appropriate scaling, it is possible to ``simulate'' money using a cake.

\begin{theorem}[\citealp{StroWoo85}]
\label{thm:simMoney}
Let $n \geq 2$, and let $(u_i)_{i \in N}$ be utility functions defined over the cake $C = [0, 1]$, with $u_i(C) = 1$ for all~$i \in N$.
Then, for any $\alpha \in [0, 1]$, there exists a set $I_\alpha \subseteq C$ such that $I_\alpha$ is the union of at most $n$ intervals and satisfies $u_i(I_\alpha) = \alpha$ for every~$i \in N$.
\end{theorem}

The above result enables us to apply techniques from the fair division with subsidies literature to establish the existence of \EFM allocations.
We make this connection more precise below.

\begin{theorem}
\label{thm:SuffCond}
In the mixed indivisible goods and chores setting, if there exists an allocation of the indivisible items that is both \EFone and \EFable, then an \EFM allocation exists with cake.
\end{theorem}

\begin{proof}
Let $A = (A_1, \dots, A_n)$ be an indivisible-items allocation that is both \EFable and \EFone.
We will show that there exists an allocation of the cake that preserves the underlying discrete allocation~$A$ and yields a combined allocation~$\A = \{(A_i, C_i)\}_{i \in N}$ of indivisible items and cake that satisfies \EFM.
Recall that under the scaling, each agent values the entire cake either at~$1$ or~$0$.
Let $N^+ \coloneqq \{i \in N \mid u_i(C) = 1\} $ and $N^0 \coloneqq N \setminus N^+$; note that for every~$j \in N^0$, we have $u_j(C) = 0$.

\begin{claim}
\label{claim:EFM-money=EFM-cake}
If there exists an \EFM allocation with money, $\{(A_i, p_i)\}_{i \in N^+}$, among the agents in~$N^+$ such that $\sum_{i \in N^+} p_i = 1$, then there exists a combined allocation $\A = \{(A_i, C_i)\}_{i \in N}$ that is \EFM.
\end{claim}

We make use of the result of \citet{StroWoo85} to prove the claim.

\begin{proof}[Proof of \Cref{claim:EFM-money=EFM-cake}]
For each $p_i > 0$ with~$i \in N^+$, we iteratively apply \Cref{thm:simMoney} to the cake~$C$ to generate pieces~$C_i$'s such that $u_j(C_i) = p_i$ for all~$j \in N^+$.
If at the end of this process, the cake~$C$ is not fully allocated, then the remaining cake must be valued at zero by every agent~$i \in N^+$,
\[
u_i \left( C \setminus \bigcup_{j \in N^+ \colon p_j > 0} C_j \right) = 1 - \sum_{j \in N^+ \colon p_j > 0} p_j = 0.
\]
We allocate the remaining cake to any agent in~$N^+$.
At this stage, the entire cake is allocated.
Consider now the full allocation of the indivisible items and the cake:
\[
\A = \begin{cases}
(A_i, C_i) & \text{ if } i \in N^+; \\
(A_i, \emptyset) & \text{ if } i \in N^0.
\end{cases}
\]

We now proceed to show allocation~$\A$ satisfies \EFM.
Recall that the underlying indivisible-items allocation~$A$ is \EFone.
For any agent~$j \in N^0$, since they value the entire cake at zero, any measurable subset of the cake is also valued at zero.
Thus, for these agents, \EFM reduces to requiring that the allocation of indivisible items be \EFone, a condition that the discrete allocation~$A$ satisfies.

For any agent~$j \in N^+$, if they envy another agent~$k \in N^0$, then \EFM holds because agent~$k$ does not receive any cake and the underlying discrete allocation~$A$ is \EFone.

Finally, for any pair of agents $i, j \in N^+$, if~$i$ is envy-free toward~$j$ in the allocation $\{(A_i, p_i)\}_{i \in N^+}$, then the envy-free relation holds in~$\A$ since $u_i(A_i) + u_i(C_i) = u_i(A_i) + p_i \geq u_i(A_j) + p_j = u_i(A_j) + u_i(C_j)$.
On the other hand, if~$i$ envies~$j$ in $\{(A_i, p_i)\}_{i \in N^+}$, then $p_j = 0$.
Since the underlying discrete allocation~$A$ is \EFone, the \EFM condition is satisfied for~$\A$.
\end{proof}

\Cref{claim:EFM-money=EFM-cake} establishes that, for an \EFone and \EFable allocation~$A$, it suffices to find payments among agents~$N^+$ that make them \EFM with money.
We now show how to construct such payments.

Note that the allocation~$A$ restricted to~$N^+$ is envy-freeable: any positive-weight cycle among the agents in~$N^+$ would also be a positive-weight cycle in the full envy graph~$\mathcal{G}_A$, contradicting the \EFable property of $A$.
Thus, no such cycle exists.

Let~$\widetilde{\mathcal{G}}$ be the envy-graph of allocation~$A$ restricted to agents~$N^+$, i.e., the subgraph of~$\mathcal{G}_A$ induced by the vertex set~$N^+$.
For each~$i \in N^+$, let~$\ell(i)$ denote the weight of the heaviest path starting from node~$i$ in~$\widetilde{\mathcal{G}}$.
Since $\widetilde{\mathcal{G}}$ has no positive-weight cycles, $\ell(i)$ is well-defined.
By Theorem~2 of \citet{HalpernSh19}, assigning each agent~$i \in N^+$ a payment $q_i = \ell(i)$ yields an envy-free allocation with money among agents in $N^+$.

If $\sum_{i \in N^+} q_i \leq 1$, set $p_i = q_i + \frac{1- \sum_{i \in N^+} q_i}{|N^+|}$ for each $i\in N^+$.
Then $\{(A_i, p_i)\}_{i \in N^+}$ remains envy-free and satisfies $\sum_{i \in N^+} p_i = 1$, establishing \EFM.

If instead $\sum_{i \in N^+} q_i > 1$, we distribute the unit amount of money while preserving \EFM.
Let $q^{(1)} > q^{(2)} > \cdots > q^{(k)} = 0$ denote the \emph{distinct} payment values among the $(q_i)_{i \in N^+}$, and for each~$t$, let $N^t = \{i \in N^+ \mid q_i = q^{(t)}\}$.
We proceed iteratively to build up the desired payments:
\begin{itemize}
\item In round~1, allocate an amount of money worth $q^{(1)} - q^{(2)}$ to each agent in~$N^1$.
\item In round~2, allocate an amount of money worth $q^{(2)} - q^{(3)}$ to each agent in $N^1 \cup N^2$.
\item Continue similarly, giving in round~$r$ an amount $q^{(r)} - q^{(r+1)}$ to each agent in~$\bigcup_{t = 1}^r N^t$, until the total allocated amount reaches exactly~$1$.

If, in some round $r$, the remaining money is insufficient to allocate the full increment $q^{(r)} - q^{(r+1)}$, we instead divide the remaining money equally among agents in~$\bigcup_{t = 1}^r N^t$; denote by~$\widehat{q}$ the amount each agent receives.
Clearly, the total money distributed equals~$1$.
\end{itemize}

Let $\{(A_i, p_i)\}_{i \in N^+}$ be the resulting allocation with payments.
By construction, we have that $\sum_{i \in N^+} p_i = 1$.
We now show $\{(A_i, p_i)\}_{i \in N^+}$ satisfies \EFM with money.
If $p_i, p_j > 0$, then $p_i - p_j = q_i - q_j$ by construction, and hence $i$ and $j$ are envy-free towards each other since the allocation $\{(A_i, q_i)\}_{i \in N^+}$ is envy-free.
Suppose $p_j > 0$ and $p_i = 0$; we must show that $i$ is envy-free towards $j$.
Since $p_j > 0$ and $p_i = 0$, we have $q_j > q_i$, meaning $j \in N^a$ and $i \in N^b$ with $a \leq r < b$, where $r$ is the final round in which money is allocated.
Then,
\[
q_j - q_i = q^{(a)} - q^{(b)} = \sum_{t = a}^{b - 1} \left( q^{(t)} - q^{(t+1)} \right) \geq \sum_{t = a}^{r - 1} \left( q^{(t)} - q^{(t+1)} \right) + \widehat{q} = p_j.
\]
Recall that $\{(A_i, q_i)\}_{i\in N^+}$ is envy-free; thus, we have $u_i(A_i) + q_i \geq u_i(A_j) + q_j$.
Agent~$i$ is envy-free towards~$j$ because $u_i(A_i) + p_i = u_i(A_i) \geq u_i(A_j) + (q_j - q_i) \geq u_i(A_j) + p_j$.
Since the discrete allocation~$A$ is \EFone and no agent envies another who received a positive payment, $\{(A_i, p_i)\}_{i \in N^+}$ satisfies \EFM with money.
Applying \Cref{claim:EFM-money=EFM-cake}, we obtain an \EFM allocation with cake.
\end{proof}

\Cref{thm:SuffCond} establishes an interesting connection between solution concepts, i.e., envy-freeness, \EFone, \EFability, and \EFM, investigated in two closely related but separate lines of research.

With only indivisible goods, an \EFone and \EFable allocation always exists~\citep{BrustleDiNa20}.
Therefore, the proof of \Cref{thm:SuffCond} provides a novel approach to compute an \EFM allocation for agents with additive preferences over a set of indivisible goods and a cake.
Our method exhibits different properties compared with the algorithm devised by \citet{BeiLiLi21}.

On the one hand, \citeauthor{BeiLiLi21}'s algorithm indeed produces an allocation satisfying a stronger envy-freeness relaxation and is able to handle monotone utilities over the indivisible goods~\citep[see,][pp.~1389 ff.]{LiuLuSu24}.
On the other hand, however, as \citet[pp.~15]{BhaskarSrVa21} pointed out, \citeauthor{BeiLiLi21}'s framework fails in the setting involving indivisible chores and a cake.
In more detail, \citeauthor{BeiLiLi21}'s algorithm relies crucially on resolving envy cycles among agents so that the algorithm can continue allocating more cake to the current allocation.
The challenge is that with indivisible chores, we can no longer resolve \emph{any} envy cycle.

Our method presented in the proof of \Cref{thm:SuffCond} \emph{preserves} the underlying discrete allocation.
Put differently, we are able to hand out pieces of cake to agents without permuting agents' bundles.
This property makes our method naturally fit to tackle the said challenge, even in the more general setting with both indivisible goods and chores, as long as an \EFone and \EFable indivisible-items allocation always exists.\footnote{Combining \Cref{prop:ind-chores:EF1+EFable} and \Cref{thm:SuffCond}, it follows immediately that with indivisible chores and a cake, an \EFM allocation always exists.}

\section{Existence of EFM}
\label{sec:EF1+EFable}

We are now ready to present the main results of our paper.

\begin{theorem}
\label{thm:EFM-existence}
For indivisible items and cake, an \EFM allocation always exists.
\end{theorem}

We will prove \Cref{thm:EFM-existence} by showing that given any instance of indivisible goods and chores, an \EFone and \EFable allocation always exists, and hence by \Cref{thm:SuffCond}, the existence of an \EFM allocation is guaranteed.

\begin{theorem}
\label{thm:EF1+EFable-existence}
For indivisible goods and chores, an \EFone and \EFable allocation always exists.
Moreover, such an allocation can be computed in polynomial time.
\end{theorem}

Although we can compute an \EFone and \EFable allocation in polynomial time, we would like to remark that we only focus on the existence result of \EFM.
This is because our cake to money reduction described in \Cref{sec:cake-to-money} does not admit a finite protocol under the Robertson-Webb model and our algorithm uses it as a subroutine.\footnote{\Cref{thm:simMoney} provides a cake division that is stronger than perfect division, which cannot be implemented with a finite number of queries in the Robertson-Webb model, even if there are only two agents~\citep{RobertsonWe98}.}

The remainder of this section is dedicated to proving \Cref{thm:EF1+EFable-existence}, which in and of itself is of independent interest.

\subsection*{Bundling Plus Iterative Matching}

We first explain the high-level approach behind \Cref{thm:EF1+EFable-existence}.
As the technique of applying iterative matching has obtained positive results in settings with only goods~\citep{BrustleDiNa20} or chores (\Cref{prop:ind-chores:EF1+EFable}), it is therefore the immediate candidate tool for us to show \EFone and \EFability in the presence of mixed goods and chores.
The idea of applying iterative matching, however, fails without bundling items together.
This can be seen from a simple instance with two items and two agents who have an identical utility over the two items: one item is of value~$+1$ while the other item is of value~$-1$.
The only way to achieve both \EFone and \EFability is to bundle the two items together and allocate the bundle to any agent.

The main idea of the proof is to carefully \emph{bundle} certain items together and treat each resulting bundle as a single \emph{meta-item}.
It is also key to us that the way we bundle items should be consistent with how the iterative matching algorithm allocates meta-items to the agents because both \EFone and \EFability will be ensured by the property of iterative matching.

In more detail, this bundling step yields a collection of pairwise disjoint sets $J = \{S_1, \dots, S_r\}$, where each $S_i \subseteq M$.
After bundling, we construct two bipartite graphs whose bipartitions correspond to the agent set~$N$ and the bundled item set~$J$ (each element $S \in J$ represents a node):
\begin{itemize}
\item bipartite graph $H[N, J]$, designed to handle objective chores (bundled or otherwise), and
\item bipartite graph $G[N, J]$, designed to handle subjective goods (also bundled or otherwise).
\end{itemize}
For each edge $(i, S)$, where $S \in J$ corresponds to a bundle, the edge weight is given by $u_i(S)$.
The graph structures of~$G$ and~$H$ differ slightly to capture the distinct nature of goods and chores.
Finally, we run the \emph{Iterative Maximum-Weight Matching} (\IMWM) algorithm on~$G$ and the \emph{Iterative Maximum-Weight Perfect Matching} (\IMWPM) algorithm on~$H$ to determine the final allocation.
The key challenge is to choose the bundling of items such that running the above algorithms yields an allocation that is both \EFable and \EFone.

\medskip
We now define several useful terms.

\begin{definition}[Meta-good and interest set]
\label{def:meta-good-interest-set}
We refer to a non-empty subset~$M' \subseteq M$ of indivisible items as a \emph{meta-good} if $u_i(M') \geq 0$ for some agent~$i \in N$.

Moreover, we will refer to the set of agents who derive non-negative utility from the meta-good~$M'$, i.e., $\{i \in N \mid u_i(M') \geq 0\}$, as its \emph{interest set}.
\end{definition}

Note that a meta-good may be a singleton.
In our algorithms and arguments, we treat each meta-good as an \emph{individual} item rather than as a set of items.
A natural restriction on meta-goods, when reasoning about \EFone, is the following property.

\begin{definition}[Chore-maximality]
\label{def:chore-maximality}
A non-empty subset~$M' \subseteq M$ of indivisible items is said to be \emph{\choreMaximal} if $u_i(M') \geq 0$ for some~$i \in N$ (i.e., $M'$ is a meta-good) and for any objective chore~$c \in M \setminus M'$, we have $u_j(M' \cup \{c\}) < 0$ for all~$j \in N$.
\end{definition}

\subsubsection*{Initial Bundling}

\begin{algorithm}[t]
\caption{Iterative Item Merging}
\label{alg:itemmerge}
\DontPrintSemicolon

\KwIn{Agents~$N$, indivisible items~$M$, and agents' utilities~$(u_i)_{i \in N}$.}
\KwOut{A set of meta-goods $\{M_1, \dots, M_\ell\}$ and a set of objective chores~$Z$.}

Let $U \subseteq M$ be the set of subjective goods in~$M$ and $Z = M \setminus U$ be the set of objective chores.

\While{there exists an agent~$i \in N$ with $\deg(i) \geq 2$ in $R[N, U]$}{ \label{alg:itemmerge:deg>=2-condition}
	Select an arbitrary agent~$i \in N$ with $\deg(i) \geq 2$.\;
	Let $S = \{g \in U \mid (i, g) \in E\}$. \tcp*{Neighbourhood items of~$i$.}
	Create a new node~$S^*$ representing~$S$. \tcp*{Merge neighbourhood items~$S$ into a meta-good.}
	\lForEach{agent~$j \in N$}{
		Add an edge $(j, S^*)$ to~$E$ if $u_j(S^*) \geq 0$.
	}
	Remove all edges incident to items in~$S$ from~$E$.\;
	Update item set $U \gets (U \setminus S) \cup \{S^*\}$.\; \label{alg:itemmerge:deg>=2-end}
}

Let $\{M_1, \dots, M_\ell\}$ be the set of $\ell$ meta-goods in $R[N, U]$. \tcp*{Each node in~$U$ is a meta-good.}

\While(\tcp*[f]{Chore maximality.}){$|Z| > 0$ \textup{\bfseries and} $\exists j \in [\ell], i \in N, c \in Z$ such that $u_i(M_j \cup \{c\}) \geq 0$}{ \label{alg:itemmerge:choreMax-condition}
	$M_j \gets M_j \cup \{c\}$\;
	$Z \gets Z \setminus \{c\}$\; \label{alg:itemmerge:choreMax-end}
}

\Return{A set of meta-goods $\{M_1, \dots, M_\ell\}$ and objective chores~$Z$}
\end{algorithm}

We begin with an initial bundling that forms the foundation for subsequent bundling steps, depending on the structure of the instance. Broadly speaking, this initial bundling guides later bundling operations and the branching of the instance into different cases.

Our initial bundling of the items aims to partition~$M$ into a set of objective chores and meta-goods each of which satisfies \choreMaximality.
We additionally require that each agent values at most one meta-good non-negatively.
This is achieved by \Cref{alg:itemmerge}.

\Cref{alg:itemmerge} takes as input an instance $\langle N, M, (u_i)_{i \in N} \rangle$ and iteratively merges items into meta-goods followed by adding chores appropriately to satisfy \choreMaximality.
Let $U \subseteq M$ be the set containing all subjective goods in~$M$ and let $Z = M \setminus U$ be the set of (current) objective chores.
Consider the weighted bipartite graph $R[N, U] = (N \cup U, E)$, where there is an edge~$(i, g) \in E$ for some~$i \in N$ and~$g \in U$ if $u_i(g) \geq 0$, with edge weight set to~$u_i(g)$.
The first \verb|while|-loop (\crefrange{alg:itemmerge:deg>=2-condition}{alg:itemmerge:deg>=2-end}) of \Cref{alg:itemmerge} identifies an agent~$i \in N$ whose degree is at least~$2$ in~$R[N, U]$ and merges agent~$i$'s neighbourhood items $S = \{g \in J \mid (i, g) \in E\}$ into a meta-good~$S^*$.
We then update the bipartite graph by removing~$S$ from the item set~$U$ and all edges incident to items in~$S$ from~$E$ as well as by adding the new node~$S^*$ to the item set~$U$ and edges~$(j, S^*)$ to~$E$ if $u_j(S^*) \geq 0$.
The \verb|while|-loop terminates when each agent values at most one node in~$U$ non-negatively.
In each iteration, $|U|$ strictly decreases; hence, the \verb|while|-loop will terminate.

When the first \verb|while|-loop (\crefrange{alg:itemmerge:deg>=2-condition}{alg:itemmerge:deg>=2-end}) terminates, assume that we have $\ell$ meta-goods, denoted as $M_1, M_2, \dots, M_\ell$.
Note that a meta-good $M_j \in \{M_1, M_2, \dots, M_\ell\}$ may be a singleton.
Recall that $Z$ contains the set of objective chores.
The second \verb|while|-loop (\crefrange{alg:itemmerge:choreMax-condition}{alg:itemmerge:choreMax-end}) of \Cref{alg:itemmerge} further merges objective chores to the $\ell$ meta-goods so that each updated meta-good could also satisfy \choreMaximality.
In more detail, in each iteration of \crefrange{alg:itemmerge:choreMax-condition}{alg:itemmerge:choreMax-end}, we identify a meta-good~$M_j$ and an objective chore~$c \in Z$ such that $u_i(M_j \cup \{c\}) \geq 0$ for some~$i \in N$, merge~$c$ into~$M_j$, and remove~$c$ from~$Z$.
Clearly, the updated~$M_j$ remains a meta-good and since we added an objective chore, we maintain the invariant that each agent values at most one meta-good non-negatively.
Since $|Z|$ strictly decreases in each iteration, the \verb|while|-loop will terminate.
At its termination, either $|Z| = 0$ or all meta-goods~$M_j \in \{M_1, \dots, M_\ell\}$ satisfy \choreMaximality.
Recall from \Cref{def:meta-good-interest-set}: for any meta-good~$M_j$, its \emph{interest set}, denoted as $T_j \coloneqq \{i \in N \mid u_i(M_j) \geq 0\}$, is the set of agents who derive non-negative utility from it.
To summarize, we have the following lemma.

\begin{lemma}
\label{lem:IterMerge}
Given an instance $\langle N, M, (u_i)_{i \in N} \rangle$, let $\{M_1, \dots, M_\ell\}$ be the set of meta-goods outputted by \Cref{alg:itemmerge}.
We have the following two properties.
\begin{itemize}
\item When $|Z| > 0$, each meta-good $M_j \in \{M_1, \dots, M_\ell\}$ is \choreMaximal.

\item The interest sets are pairwise disjoint, i.e., $T_j \cap T_k = \emptyset$ for all $j \neq k$.
\end{itemize}
\Cref{alg:itemmerge} runs in polynomial time.
\end{lemma}

Given any instance $\langle N, M, (u_i)_{i \in N} \rangle$, we apply \Cref{alg:itemmerge} and let $\{M_r\}_{r \in [\ell]}$ denote the set of meta-goods produced.
We distinguish our analyses based on the number of remaining objective chores, defined as $Z = M \setminus \bigcup_{r = 1}^\ell M_r$.
The arguments required to establish the existence of \EFone and \EFable allocations differ depending on whether $|Z| \geq n$ or $|Z| \leq n - 1$.
At a high level, these case distinctions are necessary for several reasons.
When $|Z| \geq n$, we can always combine each meta-good with a distinct objective chore from~$Z$ such that the resulting instance consists solely of objective chores.
In contrast, when $|Z| \leq n-1$, no such combination may be possible, which introduces additional technical complications and necessitates a different bundling strategy.
In the following subsections, we analyze these two cases in \Cref{sec:EFM:>=n-chores,sec:EFM:<n-chores}, respectively.

\subsection{Case~I: At Least~$n$ Remaining Objective Chores ($|Z| \geq n$)}
\label{sec:EFM:>=n-chores}

In this case, we have a collection of meta-goods $\{M_r\}_{r \in [\ell]}$ and a set of objective chores $Z = \{c_1, \dots, c_k\}$ with $k \geq n$ outputted by \Cref{alg:itemmerge}.
Recall that each of the meta-goods in $\{M_r\}_{r \in [\ell]}$ is \choreMaximal by \Cref{lem:IterMerge}.
Hence, for any objective chore~$c \in Z$, the combined bundle $M_r \circ c$ constitutes a meta-chore, i.e., $u_i(M_r \circ c) < 0$ for all~$i \in N$.\footnote{We use the notation ``$M_r \circ c$'' to denote the operation that a meta-good~$M_r$ is \emph{attached} to an objective chore~$c$, and refer to $M_r \circ c$ as a \emph{meta-chore} in order to distinguish it from a singleton chore.
In our argument presented later, given a meta-chore $M_r \circ c$, we may \emph{detach} the meta-good~$M_r$ from the singleton chore~$c$.
We thus decide to use a different notation for a meta-chore from that for a meta-good.}
Furthermore, by \Cref{lem:IterMerge}, $T_j$'s are pairwise disjoint.
Together with $\bigcup_{j \in [\ell]} T_j \subseteq N$, it follows that $\ell \leq n \leq k$.
It means that each meta-good can be paired with a distinct objective chore to form a meta-chore.

\paragraph{Pairing Configuration.}
Let $\Inj([\ell], [k])$ denote the set of all injective functions $\varphi \colon [\ell] \to [k]$ (i.e., one-to-one assignments of the $\ell$ meta-goods $\{M_r\}_{r \in [\ell]}$ to $\ell$ distinct chores in $Z$).
For each $\varphi \in  \Inj([\ell], [k])$, define the corresponding set of meta-chores, which forms a partition of the item set $M$, as
\[
J^\varphi = \left( \bigcup_{i \in [\ell]} M_i \circ c_{\varphi(i)} \right) \cup \left( Z \setminus \bigcup_{j \in [\ell]} c_{\varphi(j)} \right).
\]
The first term represents the $\ell$ meta-chores obtained by combining each of the $\ell$ meta-goods with a distinct objective chore, while the second term consists of the remaining $k - \ell$ unattached chores that are singletons.
Note that for any $\varphi \in \Inj([\ell], [k])$, the cardinality satisfies $|J^\varphi| = |Z|$.

We now introduce dummy items so that the total number of the extended set of items is a multiple of~$n$.
Let $T = \lceil k/n \rceil$ and $r = T \cdot n - k$.
We introduce $r$ dummy items $d_1, \dots, d_r$, each valued at~$0$ by every agent.
For each $\varphi \in  \Inj([\ell], [k])$, we construct the padded instance with dummy items,
\[
\widetilde{J}^\varphi = \left( \bigcup_{a \in [r]} d_a \right) \cup \left( \bigcup_{i \in [\ell]} M_i \circ c_{\varphi(i)} \right) \cup \left( Z \setminus \bigcup_{j \in [\ell]} c_{\varphi(j)} \right).
\]
By construction, we have $|\widetilde{J}^\varphi| = T \cdot n$.

\paragraph{Roundwise Optimal Pairing.}

\begin{algorithm}[t]
\caption{Iterative Maximum-Weight Perfect Matching (\IMWPM)}
\label{alg:IMWM-chores-graph}
\DontPrintSemicolon

\KwIn{Graph $H[N, \widetilde{J}^\varphi]$ with $|\widetilde{J}^\varphi| = T \cdot n$}
\KwOut{An allocation $A = (A_1, \dots, A_n)$.}

$A_i \gets \emptyset$ for all~$i \in N$.\;
$t \gets 1$; $J_1 \gets \widetilde{J}^\varphi$\;
\While{$J_t \neq \emptyset$}{
    Compute a maximum-weight perfect matching $\mu^t = \{(i, \mu_i^t)\}$ in $H[N, J_t]$.\;
    $A_i \gets A_i \cup \{\mu_i^t\}$ for all~$i \in N$.\;
    $J_{t+1} \gets J_t \setminus \bigcup_{i \in N} \{\mu_i^t\}$\;
    $t \gets t + 1$\;
}

\Return{allocation~$A$}
\end{algorithm}

For a fixed pairing~$\varphi \in \Inj([\ell], [k])$, let $H[N, \widetilde{J}^\varphi] = (N \cup \widetilde{J}^\varphi, E)$ be the complete bipartite graph between the agents and the objects.
Here, each~$h \in \widetilde{J}^\varphi$ is either a dummy item, a meta-chore of the form $M_r \circ c_{\varphi(r)}$ for some $r \in [\ell]$, or a singleton chore.
For each~$i \in N$ and~$h \in \widetilde{J}^\varphi$, edge $(i, h)$ has weight $u_i(h)$.
For $\widehat{J}\subseteq\widetilde{J}^\varphi$, write $H[N,\widehat{J}]$ for the induced subgraph.

Given a fixed pairing, \Cref{alg:IMWM-chores-graph} assigns one object to each agent in every round using a maximum-weight perfect matching among the remaining objects.
Let $\mu^t = \{(i, \mu_i^t) \colon i \in N\}$ denote the matching in round~$t$, and let $\val(\mu^t) = \sum_{i \in N} u_i(\mu_i^t)$.

We will construct the pairing and the matchings simultaneously.
At each round, we maximize $\val(\mu^t)$ over assignments consistent with the previously fixed matchings and attachments and extendable to a complete pairing.
We call this choice \emph{roundwise optimal}; ties in the current round may be broken arbitrarily before proceeding to the next round.
The resulting pairing is denoted by~$\varphi^*$.
\Cref{lem:caseI-imwpm} will show that the recorded rounds are also an execution of \IMWPM for this final pairing.

\paragraph{Polynomial-Time Implementation of the Roundwise Optimal Pairing.}
We implement this sequential choice using a maximum-profit flow network, adapting the construction of \citet{LuMaSu26}.

Let $Z_t \subseteq Z$ be the unallocated singleton chores and $B_t \subseteq [\ell]$ the indices of the unattached meta-goods at the beginning of round~$t$.
Initially, $Z_1 = Z$ and $B_1 = [\ell]$.
We maintain $|Z_t| \geq |B_t|$ and set
\begin{equation}
\label{eq:q_t}
q_t = \begin{cases}
n - r, & t = 1, \\
n, & t > 1,
\end{cases}
\qquad
b_t = |Z_t| - |B_t|.
\end{equation}
Here, $q_t$ is the number of real chores selected in round~$t$.
The algorithm allocates all $r$ dummy items in the first round by construction, assigning a real object to each of the other $n-r$ agents.
Every subsequent round assigns a real object to every agent.
This convention is consistent with optimizing the current round: under any complete pairing, replacing a selected real object with an unused dummy strictly increases the matching value.
A \emph{candidate round} assigns $q_t$ distinct chores to distinct agents and attaches $p$ distinct meta-goods to distinct selected chores.
When
\begin{equation}
\label{eq:caseI-extendability}
|Z_t| - q_t \geq |B_t| - p, \qquad \text{equivalently,} \qquad q_t - p \leq b_t,
\end{equation}
it leaves enough singleton chores to attach all remaining meta-goods, meaning that this round is \emph{extendable}.

For each round, construct a directed network $F_t$ with source~$s$, sink~$z$, one node for each~$c \in Z_t$, two nodes $i^-, i^+$ for each~$i \in N$, one node~$v_j$ for each~$j \in B_t$, and an auxiliary node~$v_0$; see \Cref{fig:round-flow-network}.
Its arcs, capacities, and profits are as follows; the indicated indices range over $c \in Z_t$, $i \in N$, and $j \in B_t$.
\[
\begin{array}{lll c lll}
\text{Arc} & \text{Capacity} & \text{Profit} & \hspace{5em} & \text{Arc} & \text{Capacity} & \text{Profit} \\ \hline
s \longrightarrow c & 1 & 0 & & i^+ \longrightarrow v_j & 1 & u_i(M_j) \\
c \longrightarrow i^-& 1 & u_i(c) & & i^+ \longrightarrow v_0 & 1 & 0 \\
i^- \longrightarrow i^+ & 1 & 0 & & v_j \longrightarrow z & 1 & 0 \\
& & & & v_0 \longrightarrow z & b_t & 0
\end{array}
\]
Given $F_t$, we compute an integral maximum-profit flow of \emph{exactly} $q_t$ units.
A flow path through chore~$c$, agent nodes $i^-,i^+$, and~$v_j$ assigns $M_j\circ c$ to agent~$i$; a path using~$v_0$ instead assigns only~$c$ to~$i$.
An agent carrying no flow receives a distinct dummy item.
This last situation occurs only in the first round.
The pseudocode can be found in \Cref{alg:caseI-flow}.

\begin{figure}[t]
\centering
\begin{tikzpicture}[
>=latex,
scale=1.14,
transform shape,
vertex/.style={circle, draw, line width=0.45pt, minimum size=8.5mm, inner sep=0pt, font=\small},
arc/.style={->, line width=0.6pt},
fam/.style={->, line width=0.42pt, densely dashed}, 
pairarc/.style={->, line width=0.5pt, dashed}, 
bypass/.style={->, line width=0.55pt},
title/.style={font=\footnotesize\bfseries, align=center}
]
\node[title] at (0,2.45) {source};
\node[vertex] (s) at (0,0) {$s$};

\node[title] at (2,2.45) {real chores};
\node[vertex] (c1) at (2,1.5) {$c_1$};
\node[vertex] (c2) at (2,0) {$c_2$};
\node at (2,-.7) {\(\vdots\)};
\node[vertex] (c3) at (2,-1.5) {$c_{|Z_t|}$};

\draw[arc] (s) -- (c1);
\draw[arc] (s) -- (c2);
\draw[arc] (s) -- (c3);

\node[title] at (4.5,2.45) {agent input};
\node[vertex] (r1m) at (4.5,1.5) {$i_1^-$};
\node[vertex] (r2m) at (4.5,0) {$i_2^-$};
\node at (4.5,-.7) {\(\vdots\)};
\node[vertex] (r3m) at (4.5,-1.5) {$i_n^-$};

\foreach \c in {c1,c2,c3} {
    \foreach \r in {r1m,r2m,r3m} {
        \draw[fam] (\c) -- (\r);
    }
}

\node[title] at (7,2.45) {agent output};
\node[vertex] (r1p) at (7,1.5) {$i_1^+$};
\node[vertex] (r2p) at (7,0) {$i_2^+$};
\node at (7,-.7) {\(\vdots\)};
\node[vertex] (r3p) at (7,-1.5) {$i_n^+$};

\draw[arc] (r1m) -- (r1p);
\draw[arc] (r2m) -- (r2p);
\draw[arc] (r3m) -- (r3p);

\node[title] at (9.5,2.45) {meta-goods};
\node[vertex] (v0) at (10,1.5) {$v_0$};
\node[vertex] (v1) at (9.5,0) {$v_1$};
\node at (9.5,-.7) {\(\vdots\)};
\node[vertex] (vBt) at (9.5,-1.5) {$v_{|B_t|}$};

\foreach \r in {r1p,r2p,r3p} {
    \foreach \m in {v1,vBt} {
        \draw[pairarc] (\r) -- (\m);
    }
}

\draw[bypass] (r1p) to (v0);
\draw[bypass] (r2p) to (v0);
\draw[bypass] (r3p) to (v0);

\node[title] at (11.5,2.45) {sink};
\node[vertex] (t) at (11.5,0) {$z$};

\draw[arc] (v0) to (t);
\draw[arc] (v1) -- (t);
\draw[arc] (vBt) -- (t);
\end{tikzpicture}
\caption{The maximum-profit flow network for round~$t$.
Flow through a meta-good node~$v$ pairs the selected chore with that meta-good; flow through $v_0$ leaves the selected chore unpaired.}
\label{fig:round-flow-network}
\end{figure}

\begin{lemma}
\label{lem:caseI-flow}
Integral flows of value~$q_t$ in $F_t$ encode precisely the candidate rounds satisfying \Cref{eq:caseI-extendability}, up to the identities of dummy items.
The flow profit equals the matching value of the encoded round.
\end{lemma}

\begin{proof}
The unit capacities ensure that each selected chore and each attached meta-good is used once, and that each agent receives at most one real object.
If $p$ meta-goods are attached, exactly $q_t-p$ units pass through $v_0$.
The capacity of $v_0\longrightarrow z$ therefore imposes \Cref{eq:caseI-extendability}.
Conversely, any such candidate round gives a feasible flow by routing each assigned object along its corresponding path.
Since the network is acyclic, every integral flow decomposes into these paths.
Their profits are $u_i(c)$ and $u_i(c)+u_i(M_j)$, respectively, exactly the utilities of the assigned objects.
\end{proof}

\begin{algorithm}[t]
\caption{Case~I: Roundwise Optimal Pairing and Allocation}
\label{alg:caseI-flow}
\DontPrintSemicolon

\KwIn{Meta-goods $M_1, \dots, M_\ell$ and chores~$Z$ from \Cref{alg:itemmerge}, with $k = |Z|\geq n$.}
\KwOut{An \EFone and \EFable allocation of the original items.}

$T \gets \lceil k/n \rceil$; $r \gets T \cdot n - k$; create $r$ zero-valued dummy items.\;
$Z_1 \gets Z$; $B_1 \gets [\ell]$; $A_i \gets \emptyset$ for every~$i \in N$.\;
\For{$t = 1, \dots, T$}{
	Set $b_t = |Z_t| - |B_t|$ and
		$q_t = \begin{cases}
		n - r, & t = 1, \\
		n, & t > 1.
		\end{cases}$\;
	Build $F_t$ (see \Cref{fig:round-flow-network}) and compute an integral maximum-profit flow of value~$q_t$.\;
	Decode the flow into objects~$\mu_i^t$, assigning a distinct unused dummy item to every agent carrying no flow.
	Add the original items represented by~$\mu_i^t$ to~$A_i$ for every~$i \in N$; ignore dummy items.\;
	Remove selected chores from~$Z_t$ and attached meta-good indices from~$B_t$ to obtain $Z_{t+1}$ and $B_{t+1}$.\;
}
\Return{$A = (A_1, \dots, A_n)$}
\end{algorithm}

In the remainder of this subsection, we will show that \Cref{alg:caseI-flow} returns an allocation that is both \EFone and \EFable.
We first verify that every round is feasible and that the algorithm allocates all original items.
Initially, $|Z_1|=k\geq\ell=|B_1|$ and $1\leq q_1=n-r\leq\min\{n,k\}$.
Suppose that a round begins with $|Z_t|\geq|B_t|$ and $1\leq q_t\leq\min\{n,|Z_t|\}$.
Select any $q_t$ chores and attach $p=\max\{0,q_t-b_t\}$ distinct meta-goods to distinct selected chores.
Since $q_t\leq|Z_t|$ and $b_t=|Z_t|-|B_t|\geq0$, we have $p\leq\min\{q_t,|B_t|\}$.
Assigning the resulting objects to distinct agents gives an extendable candidate round and hence, by \Cref{lem:caseI-flow}, a feasible flow of value~$q_t$.

For the round selected by the algorithm, \Cref{eq:caseI-extendability} gives $|Z_{t+1}|\geq|B_{t+1}|$.
The first round leaves $(T-1)n$ real chores, and each subsequent round removes exactly~$n$.
Thus, inductively, $|Z_t|=(T-t+1)n\geq n=q_t$ for every $2\leq t\leq T$, so all $T$ rounds are feasible.
At their completion, $Z_{T+1}=\emptyset$ and the invariant implies $B_{T+1}=\emptyset$.
Each chore has been assigned once and each meta-good attached once, so all original items have been allocated.
The recorded attachments define an injective map~$\varphi^*$.

\begin{lemma}
\label{lem:caseI-imwpm}
Let $\varphi^*$ be the pairing produced by \Cref{alg:caseI-flow}, and let $(\mu^1,\dots,\mu^T)$ be its recorded matchings, including the assignments of dummy items.
For each round~$t$, define
\[
R_t = \widetilde{J}^{\varphi^*}
\setminus \{\mu_i^s : i\in N,\ 1\leq s<t\},
\]
which contains the objects under this fixed pairing that remain after the algorithm's first $t-1$ rounds.
Then $\mu^t$ is a maximum-weight perfect matching in $H[N,R_t]$.
Consequently, the recorded sequence is an execution of $\IMWPM(H[N,\widetilde{J}^{\varphi^*}])$, with suitable tie-breaking.
\end{lemma}

\begin{proof}
Fix a round~$t$ and the final pairing~$\varphi^*$.
Every chore in~$Z_t$ underlies exactly one real object in~$R_t$.
In the final pairing, every meta-good indexed by~$B_t$ is attached to a distinct chore in~$Z_t$: chores allocated before round~$t$ are never reused, and all remaining meta-goods are attached in round~$t$ or later.
Thus, $R_t$ contains $|B_t|$ meta-chores and exactly $b_t = |Z_t|-|B_t|$ singleton chores.

Let $\nu = \{(i, \nu_i) \colon i \in N\}$ be a maximum-weight perfect matching in $H[N,R_t]$.
By chore-maximality, every real object has strictly negative utility for every agent, whereas dummy items have utility zero.
Therefore, $\nu$ uses every available dummy item; otherwise, replacing a selected real object with an unused dummy would increase its weight.
By construction of \Cref{alg:caseI-flow}, $R_t$ contains $r$ dummy items when $t=1$ and none afterwards.
Hence, $\nu$ selects exactly $q_t$ real objects, where $q_t$ is defined in \Cref{eq:q_t}.

Suppose these selected objects include $p$ meta-chores.
The other $q_t-p$ objects are singleton chores.
Since $R_t$ contains only $b_t$ singleton chores, we have $q_t-p \leq b_t$.
This is precisely the extendability condition in \Cref{eq:caseI-extendability}.
Therefore, the assignments in~$\nu$, retaining their pairings under~$\varphi^*$, form a candidate round represented by a feasible flow of value~$q_t$ in~$F_t$.
By \Cref{lem:caseI-flow}, its profit equals the weight of~$\nu$.

Since the algorithm selects a maximum-profit flow in~$F_t$, we have $\sum_{i\in N}u_i(\mu_i^t) \geq \sum_{i\in N}u_i(\nu_i)$.
Moreover, $\mu^t$ itself is a feasible perfect matching in $H[N,R_t]$: the final pairing preserves every object assigned in round~$t$.
Thus, $\mu^t$ is a maximum-weight perfect matching in that graph.
This holds for every round~$t$.
\end{proof}

The next lemma ensures that \EFone holds for original items, not merely for bundled objects.
Its proof uses the fact that every agent receives a real object in the final round.
Indeed, $q_T=n$: if $T=1$, then $k\geq n$ and $T=\lceil k/n\rceil$ force $k=n$ and $r=0$; if $T>1$, then the definition of $q_t$ directly gives $q_T = n$.

\begin{lemma}
\label{lem:RightAgent}
If $\mu_i^T = M_j \circ c$ for some~$j \in [\ell]$ and~$c \in Z$, then~$i \in T_j$, i.e., $u_i(M_j) \geq 0$.
\end{lemma}

\begin{proof}
Keep each agent's underlying real chore fixed and temporarily detach the meta-goods assigned in round~$T$.
For each such meta-good~$M_j$, choose $a_j \in \argmax_{a \in N} u_a(M_j)$.
Because $M_j$ is a meta-good, $u_{a_j}(M_j) \geq 0$, so $a_j \in T_j$.
The agents~$a_j$ are distinct because the interest sets are pairwise disjoint.
Since every agent has a real chore in this round, we can attach each~$M_j$ to the chore held by~$a_j$.

This reassignment uses exactly the same chores and meta-goods as before and changes none of the preceding rounds.
It therefore remains extendable and corresponds to a feasible flow in~$F_T$ by \Cref{lem:caseI-flow}.
The total contribution from the chores is unchanged, and the contribution from each meta-good weakly increases.
If some meta-good had been assigned to an agent who valued it negatively, its contribution would strictly increase.
The resulting flow would then have strictly greater profit, contradicting optimality of the flow chosen in round~$T$.
\end{proof}

\begin{proposition}
\label{prop:EF1+EFable:>=n-chores}
Under Case~I, \Cref{alg:caseI-flow} returns an \EFone and \EFable allocation in polynomial time.
\end{proposition}

\begin{proof}
Let $A$ be the allocation returned by the algorithm and $(\mu^1, \dots, \mu^T)$ its recorded matchings.
We interpret $u_i(\mu_j^t)$ as the utility of the original items in the object, with utility zero for a dummy item.

\medskip
\noindent\textbf{Envy-freeability.}
By \Cref{lem:caseI-imwpm}, every $\mu^t$ is a maximum-weight perfect matching in the remaining object collection for~$\varphi^*$.
For any permutation~$\sigma$ of~$N$, reassigning the objects of round~$t$ according to~$\sigma$ is another feasible matching.
It follows that
\[
\sum_{i\in N}u_i(A_i)
=\sum_{t=1}^T\sum_{i\in N}u_i(\mu_i^t)
\geq\sum_{t=1}^T\sum_{i\in N}u_i(\mu_{\sigma(i)}^t)
=\sum_{i\in N}u_i(A_{\sigma(i)}).
\]
Thus $A$ is \EFable by \Cref{thm:hs_ef}.

\medskip
\noindent\textbf{\EFone.}
To prove \EFone, fix any pair of agents~$i, j \in N$.
For every $t < T$, the object $\mu_j^{t+1}$ is still available in round~$t$.
Replacing $i$'s matched object by this object is feasible, so maximum weight gives
$u_i(\mu_i^t) \geq u_i(\mu_j^{t+1})$.
Since $q_T=n$, agent~$i$'s final-round object contains an original objective chore~$c_i^T\in Z$, which belongs to~$A_i$.
If that object is a singleton, then $u_i(\mu_i^T) - u_i(c_i^T) = 0$.
Otherwise this difference is the value of its attached meta-good, which is non-negative by \Cref{lem:RightAgent}.
Therefore,
\begin{align*}
u_i(A_i \setminus \{c_i^T\})
= \sum_{t=1}^{T-1} u_i(\mu_i^t) + u_i(\mu_i^T) - u_i(c_i^T)
\geq \sum_{t=1}^{T-1} u_i(\mu_i^t)
\geq \sum_{t=2}^{T} u_i(\mu_j^t)
\geq \sum_{t=1}^{T} u_i(\mu_j^t)
= u_i(A_j).
\end{align*}
The last inequality uses $u_i(\mu_j^1) \leq 0$.
Empty sums are zero, so the argument also covers $T=1$.
Removing the single original chore~$c_i^T$ eliminates all of agent~$i$'s envy, proving \EFone.

\medskip
\noindent\textbf{Polynomial-Time Computation.}
There are $T = \lceil k / n \rceil$ rounds, and the flow network in each round has polynomially many nodes and arcs.
In each round, negating the profits yields a minimum-cost flow problem of prescribed value $q_t \leq n$, with integral capacities.
An integral optimum can be obtained by standard minimum-cost flow algorithms in time polynomial in the input bit length.
Note also that the preceding item-merging procedure is polynomial-time as well.
\end{proof}

\subsection{Case~II: At Most~$n-1$ Remaining Objective Chores ($|Z| \leq n-1$)}
\label{sec:EFM:<n-chores}

In this case, we have a collection of meta-goods~$\{M_r\}_{r \in [\ell]}$ and a set of objective chores~$Z$ with $|Z| \leq n-1$.
Note that when $|Z| \leq n - 1$, there may not be enough chores to be combined with the meta-goods to ensure the resulting instance contains only chores.
In this case, some meta-goods must be allocated on their own.
This implies that, in order to allocate the meta-goods $\{M_r\}_{r \in [\ell]}$, we additionally require \emph{\goodMinimality}.

\begin{definition}[Good-minimality]
\label{def:good-minimality}
A non-empty subset~$M' \subseteq M$ of indivisible items is said to be \emph{\goodMinimal} if $u_i(M') \geq 0$ for some~$i \in N$ (i.e., $M'$ is a meta-good) and for all~$j \in N$ and $g \in \{o \in M' \mid u_j(o) \geq 0\}$, $u_j(M' \setminus \{g\}) \leq 0$.
\end{definition}

We show below the connection between \goodMinimality and \EFone.

\begin{lemma}
\label{lemma:good-minimality+EF1}
Consider any allocation $(A_i)_{i \in N}$ of items~$M$ and assume agent~$j$ receives a meta-good~$M'$ that is \goodMinimal.
If $u_i(A_i) \geq u_i(A_j \setminus M')$ for some~$i \in N$, then agent~$i$ is \EFone towards agent~$j$.
\end{lemma}

\begin{proof}
If $u_i(M') \leq 0$, we have $u_i(A_i) \geq u_i(A_j \setminus M') \geq u_i(A_j)$.
Thus, agent~$i$ is envy-free towards agent~$j$.
Else, $u_i(M') > 0$.
Since the meta-good~$M'$ is \goodMinimal by the statement of the lemma, there must exist a good~$g \in M'$ for agent~$i$ such that $u_i(M' \setminus \{g\}) \leq 0$.
As a result,
\[
u_i(A_i) \geq u_i(A_j \setminus M') \geq u_i((A_j \setminus M') \cup (M' \setminus \{g\})) = u_i(A_j \setminus \{g\}),
\]
meaning that agent~$i$ is \EFone towards agent~$j$, as desired.
\end{proof}

We will first consider the scenario with $|Z| = 0$ in \Cref{sec:EFM:<n-chores:=0} and design an algorithm that always produces an \EFone and \EFable allocation, which will prove useful for the scenario with $1 \leq |Z| \leq n-1$ in \Cref{sec:EFM:<n-chores:>0}.

\subsubsection{Case~II.1: $|Z| = 0$}
\label{sec:EFM:<n-chores:=0}

\begin{algorithm}[t]
\caption{Iterative Matching for Meta-Goods}
\label{alg:meta-goods}
\DontPrintSemicolon

\SetKwFunction{FMain}{\texttt{\textup{IMWM}}}
\KwIn{A set of meta-goods~$\{M_1, \dots, M_\ell\}$ and a (possibly empty) set of subjective goods~$X$, with $(\bigcup_{j \in [\ell]} M_j) \cup X = M$.}
\KwOut{An \EFone and \EFable allocation~$A = (A_1, \dots, A_n)$.}

\lForEach{$g \in X$}{
    Create the $(\ell + 1)$-th meta-good $M_{\ell + 1} \gets \{g\}$, and update $\ell \gets \ell + 1$. \label{alg:meta-goods:create-meta-goods}
}

\While(\tcp*[f]{Good-minimality}){$\exists j \in [\ell], i \in N, g \in M_j$ such that $u_i(g) \geq 0$ and $u_i(M_j \setminus \{g\}) > 0$}{ \label{alg:meta-goods:good-min-begin}
    Create the $(\ell + 1)$-th meta-good $M_{\ell + 1} \gets \{g\}$.\;
    $M_j \gets M_j \setminus \{g\}$\;
    $\ell \gets \ell + 1$\; \label{alg:meta-goods:good-min-end}
}

$J \gets \{M_1, \dots, M_\ell\}$.\;
$(A_1, \dots, A_n) \gets$ \FMain{$G[N, J]$}\;

\Return{allocation~$A$}

\BlankLine
\hrule
\BlankLine

\SetKwProg{Fn}{Function}{:}{\KwRet{}}
\Fn{\FMain{$G[N, J]$}}{ \label{alg:IMWM}
    $A_i \gets \emptyset$ for all~$i \in N$.\;
    $t \gets 1$; $J_1 \gets J$\;
    \While{$J_t \neq \emptyset$}{
        Compute a maximum-weight matching $\mu^t = \{(i, \mu_i^t)\}$ in $G[N, J_t]$, breaking ties in favour of larger cardinality size.\;
        $A_i \gets A_i \cup \mu_i^t$ for all~$i \in N$.\;
        $J_{t+1} \gets J_t \setminus \{\mu_i^t \colon i \in N,\ \mu_i^t \neq \emptyset\}$\;
        $t \gets t + 1$\;
    }

	\KwRet{Allocation A}
}
\end{algorithm}

In this subcase, we study the scenario that after \Cref{alg:itemmerge} preprocesses the given instance, there is a set of $\ell$ meta-goods and $0$ remaining objective chores.
We design \Cref{alg:meta-goods} to further process the output of \Cref{alg:itemmerge} and present its pseudocode in a more general sense so that it can be used as a subroutine for certain cases encountered later in \Cref{sec:EFM:<n-chores:>0}.
At a high level, taking as input a set of $\ell$ meta-goods $\{M_1, \dots, M_\ell\}$ and a (possibly empty) set of subjective goods~$X$, \Cref{alg:meta-goods} first obtains an updated set of meta-goods each of which also satisfies \goodMinimality in \crefrange{alg:meta-goods:create-meta-goods}{alg:meta-goods:good-min-end}, and next it applies an Iterative Maximum-Weight Matching (\IMWM) procedure to allocate the meta-goods to agents~$N$.

Denote by $M_{\ell + 1}, M_{\ell + 2}, \dots, M_{\ell + |X|}$ the meta-goods representation of the $|X|$ subjective goods from~$X$, and update $\ell \gets \ell + |X|$ accordingly.
We now have a partition of items~$M$ into meta-goods $M_1, \dots, M_\ell$, which may not all satisfy \goodMinimality.
The algorithm identifies in \cref{alg:meta-goods:good-min-begin} an~$M_j$ that violates \goodMinimality, i.e., there exists an agent~$i \in N$ and her subjective good~$g \in M_j$ such that $u_i(g) \geq 0$ and $u_i(M_j \setminus \{g\}) > 0$.
We then create the $(\ell + 1)$-th meta-good $M_{\ell + 1} \gets \{g\}$, as well as update $M_j \gets M_j \setminus \{g\}$ and $\ell \gets \ell + 1$.
Clearly, the updated~$M_j$ remains a meta-good.
Since $\ell$ strictly increases in each iteration of the \verb|while|-loop in \crefrange{alg:meta-goods:good-min-begin}{alg:meta-goods:good-min-end} and is upper bounded by $m$, the \verb|while|-loop will terminate.
At its termination, all~$M_j$'s satisfy \goodMinimality.
Let $J = \{M_1, \dots, M_\ell\}$.

Given agents~$N$ and a set $J = \{M_1, \dots, M_\ell\}$ of meta-goods, let $G[N, J] = (N \cup J, E)$ denote a weighted bipartite graph where there is an edge~$(i, M_j)$ for some~$i \in N$ and~$M_j \in J$ if $u_i(M_j) \geq 0$; the edge has weight $u_i(M_j)$.\footnote{Note that $G[N, J]$ defined here is different from $H[N, J]$ defined in \Cref{sec:EFM:>=n-chores}.
Moreover, while we are inspired by the Iterative Matching Algorithm of \citet{BrustleDiNa20} for goods, the way we construct $G[N, J]$ is different from theirs.} For any subset $\widehat{J}\subseteq J$, let $G[N,\widehat{J}]$ denote the subgraph of $G[N,J]$ induced by $N$ and $\widehat{J}$.
In each round of \IMWM, each agent is matched to at most one meta-good.
For the first round, we set $J_1 = \{M_1,\dots, M_\ell\}$.
In each round~$t$, we find a maximum-weight matching~$\mu^t$ in~$G[N, J_t]$ and agent~$i$ is matched to item~$\mu^t_i$.\footnote{Among all maximum-weight matchings, we choose one of maximum cardinality.
This tie-breaking rule guarantees progress even when all available edge weights are zero and justifies the strict negativity assertion below; the \EFone and \EFability arguments otherwise use only maximum weight.}
If agent~$i$ does not match to any item, let $\mu^t_i = \emptyset$ for the ease of expression.
If agent~$i$ is unmatched in round~$t$, then $u_i(h) < 0$ for every meta-good~$h \in J_{t+1}$.
Otherwise, adding the edge~$(i, h)$ would either increase the matching weight or preserve its weight while increasing its cardinality, contradicting the selection rule.
Put differently, each agent will only be allocated meta-goods of non-negative values for them.
The \IMWM procedure terminates when all items are matched to agents.
Whenever $J_t$ is non-empty, some edge exists because every remaining meta-good is valued non-negatively by at least one agent.
The maximum-cardinality tie-break therefore ensures that at least one meta-good is allocated in that round, even if the maximum matching weight is zero.
Hence \IMWM terminates after at most $|J|$ rounds.
We are now ready to establish the following statement.

\begin{proposition}
\label{prop:EF1+EFable:<n-chores:=0}
Taking as input a set of meta-goods $\{M_j\}_{j \in [\ell]}$ and a (possibly empty) set of subjective goods~$X$ with $X \cup \bigcup_{j \in [\ell]} M_j = M$, \Cref{alg:meta-goods} always finds an \EFone and \EFable allocation in polynomial time.
\end{proposition}

\begin{proof}
At the end of the first \verb|while|-loop (\crefrange{alg:meta-goods:good-min-begin}{alg:meta-goods:good-min-end}) of \Cref{alg:meta-goods}, each meta-good~$M_j$ in $J= \{M_1, \dots, M_\ell\}$ is \goodMinimal.
It suffices to show that the output allocation~$A = (A_1, \dots, A_n)$ of $\IMWM(G[N, J])$ satisfies both \EFone and \EFability.
Suppose that \IMWM terminates in~$T$ rounds.
For each round~$t \in [T]$, let $\mu^t$ denote the matching in this round; moreover, each agent~$i \in N$ is matched to a meta-good~$\mu^t_i$.
Note that $\mu^t_i = \emptyset$ if the agent does not match to anything.
At the end of \IMWM, each agent~$i \in N$ receives a bundle $A_i = \bigcup_{t \in [T]} \mu_i^t$.

\medskip
\noindent\textbf{Allocation~$A$ is \EFable.}
By \Cref{thm:hs_ef}, it suffices to show that the envy graph~$\mathcal{G}_A$ has no positive-weight cycles.
Take any directed cycle~$\gamma$ in~$\mathcal{G}_{A}$ and assume without loss of generality that $\gamma$ involves a sequence of nodes/agents $1, 2, \dots, k$ for some~$k \geq 2$.
We have
\begin{align*}
w_{A}(\gamma) = \sum_{(i, j) \in \gamma} w_{A}(i, j) = \sum_{(i, j) \in \gamma} \left( u_i(A_j) - u_i(A_i) \right)
&= \sum_{(i, j) \in \gamma} \sum_{t = 1}^T \left( u_i(\mu^t_j) - u_i(\mu^t_i) \right) \\
&= \sum_{(i, j) \in \gamma} \sum_{t = 1}^T w_{\mu^t}(i, j) = \sum_{t = 1}^T \sum_{(i, j) \in \gamma} w_{\mu^t}(i, j)
\end{align*}

Let $\pi_\gamma$ be the permutation of~$N$ under which $\pi_\gamma(i) = i + 1$ for each $i \in [k-1]$, $\pi_\gamma(k) = 1$, and $\pi_\gamma(i) = i$ for $i \in N \setminus [k]$.
In each round~$t$, $\mu^t$  is a maximum-weight matching; therefore, $\sum_{(i, j) \in \gamma} w_{\mu^t}(i, j) \leq 0$.
Otherwise, consider the reassignment that gives each agent~$i$ the meta-good $\mu^t_{\pi_\gamma(i)}$, with $\emptyset$ interpreted as receiving nothing.
Its total utility is strictly greater than the weight of~$\mu^t$.
This reassignment may use pairs of negative utility, which are not edges of~$G[N, J_t]$.
Discard every such pair and leave its agent unmatched; also omit empty assignments.
The remaining pairs form a feasible matching in~$G[N, J_t]$, since the reassignment gives each agent at most one meta-good and assigns each meta-good at most once.
Discarding negative-valued pairs cannot decrease total weight, so this matching has strictly greater weight than~$\mu^t$, a contradiction.
Thus, $w_A(\gamma) \leq 0$, implying that allocation~$A$ is \EFable.

\medskip
\noindent\textbf{Allocation~$A$ is \EFone.}
Observe that $u_i(\mu_i^t) \geq u_i(g)$ for any meta-good~$g \in J_{t+1}$; otherwise, replacing the edge~$(i, \mu^t_i)$ with~$(i, g)$ in~$\mu^t$ results in a matching of greater weight in~$G[N, J_t]$.
For any~$i, j \in N$,
\begin{align*}
u_i(A_i) = u_i(\{\mu_i^1, \dots, \mu_i^T\}) = u_i(\mu_i^1) + \cdots + u_i(\mu_i^{T-1}) + u_i(\mu_i^T)
&\geq u_i(\mu_i^1) + \cdots + u_i(\mu_i^{T-1}) \\
&\geq u_i(\mu_j^2) + \cdots + u_i(\mu_j^T).
\end{align*}
If $u_i(\mu_j^1) \leq 0$, clearly $u_i(A_i) \geq u_i(\mu_j^1) + u_i(\mu_j^2) + \cdots + u_i(\mu_j^T) = u_i(A_j)$, meaning that agent~$i$ is envy-free towards agent~$j$.
Else, $u_i(\mu_j^1) > 0$, meaning that $u_i(A_i) \geq u_i(A_j \setminus \mu_j^1)$.
In other words, after removing meta-good~$\mu^1_j$ from agent~$j$'s bundle, agent~$i$ would be envy-free towards~$j$.
By \Cref{lemma:good-minimality+EF1}, agent~$i$ is \EFone towards~$j$, as desired.

\medskip
\noindent\textbf{Polynomial-Time Computation.}
The splitting step increases the number of non-empty, pairwise disjoint meta-goods by one, so there are at most $m$ such steps. A violating triple $(j,i,g)$ can be found, or its absence certified, by scanning all agents and all items in the current meta-goods and testing $u_i(g)\geq0$ and $u_i(M_j)-u_i(g)>0$. These tests and the updates take polynomial time.
After splitting, $|J|\leq m$. In each matching round, the maximum-cardinality tie-break ensures that at least one meta-good is allocated whenever any remain, so there are at most $m$ rounds. Each round finds a maximum-weight matching, and among those one of maximum cardinality, in polynomial time on a graph with at most $n+m$ vertices.
All bundle utilities and edge weights are sums of rational input utilities and have polynomial bit length. Exact arithmetic and the cardinality tie-break can therefore be implemented in polynomial time, proving the claim.
\end{proof}

\subsubsection{Case~II.2: $1 \leq |Z| \leq n-1$}
\label{sec:EFM:<n-chores:>0}

\begin{algorithm}[t]
\caption{Good-Minimal and Chore-Maximal Refinement}
\label{alg:chore-max-good-min}
\DontPrintSemicolon

\KwIn{Meta-goods $\{M_j\}_{j \in [\ell]}$ and chores~$Z$ with $ 1 \leq |Z| \leq n-1$ outputted by \Cref{alg:itemmerge}.}
\KwOut{An updated set of meta-goods $\{M_1, \dots, M_{\widehat{\ell}}\}$ that are both \choreMaximal and \goodMinimal, a set of subjective goods~$X$, and an updated set of objective chores~$\widehat{Z}$.}

$X \gets \emptyset$\;
$\mathcal{M} \gets \{M_1, \dots, M_\ell\}$\;
$\widehat{Z} \gets Z$\;
$\widehat{\ell} \gets \ell$\;

\While{$\exists M_j \in \mathcal{M}, i \in N, g \in M_j$ such that $u_i(g) \geq 0$ and $u_i(M_j \setminus \{g\}) > 0$
\textup{\textbf{or}} \nonl\\
\makebox[\widthof{\textbf{while}}]{}
\textup{(}$|\widehat{Z}| \neq 0$ \textup{\textbf{and}} $\exists i \in N, S \subseteq X \cup \mathcal{M}, c \in \widehat{Z}$ such that $u_i(S \cup \{c\}) \geq 0$\textup{)}}{ \label{alg:chore-max-good-min:bundle-condition}
    \eIf(\tcp*[f]{Good minimality.}){$\exists M_j \in \mathcal{M}, i \in N, g \in M_j$ s.t.\ $u_i(g) \geq 0$ and $u_i(M_j \setminus \{g\}) > 0$}{
        $M_j \gets M_j \setminus \{g\}$\;
        $X \gets X \cup \{g\}$\;
    }(\tcp*[f]{$|\widehat{Z}| \neq 0$ and $\exists i \in N, S \subseteq X \cup \mathcal{M}, c \in \widehat{Z}$ such that $u_i(S \cup \{c\}) \geq 0$}){
        $\widehat{\ell} \gets \widehat{\ell} + 1$\;
        $M_{\widehat{\ell}} \gets S \cup \{c\}$\;
        $\mathcal{M} \gets (\mathcal{M} \setminus S) \cup \{M_{\widehat{\ell}}\}$\;
        $X \gets X \setminus S$\;
        $\widehat{Z} \gets \widehat{Z} \setminus \{c\}$\;
    }

    \While(\tcp*[f]{Merging.}){$\exists i' \in N, M_r, M_{r'} \in \mathcal{M}$ with $r \neq r'$ such that $u_{i'}(M_r) \geq 0$ and $u_{i'}(M_{r'}) \geq 0$}{ \label{alg:chore-max-good-min:merge-begin}
        $M_r \gets M_r \cup M_{r'}$ \tcp*{Merge~$M_{r'}$ into~$M_r$; get an updated meta-good~$M_r$.}
        $\mathcal{M} \gets \mathcal{M} \setminus \{M_{r'}\}$\; \label{alg:chore-max-good-min:merge-end}
    }
}

$\widehat{\ell} \gets |\mathcal{M}|$\;
Relabel all~$M_j \in \mathcal{M}$ such that we have the updated set of meta-goods $\{M_1, \dots, M_{\widehat{\ell}}\}$.\;

\Return{meta-goods $\{M_1, \dots, M_{\widehat{\ell}}\}$, subjective goods~$X$, and objective chores~$\widehat{Z}$}
\end{algorithm}

We now consider the case where \Cref{alg:itemmerge} outputs meta-goods $\{M_r\}_{r \in [\ell]}$ that are chore-maximal and a set of objective chores~$Z$ with $1 \leq |Z| \leq n-1$.
We design \Cref{alg:chore-max-good-min} to further process those outputs and intend to find a partition of~$M$ into meta-goods~$\{M_1, \dots, M_{\widehat{\ell}}\}$, a set of subjective goods~$X = \{g_1, \dots, g_s\}$, and a set of objective chores~$\widehat{Z}$ such that the following properties hold:
\begin{itemize}
\item Each element of $\{M_1, \dots, M_{\widehat{\ell}}\}$ is both \goodMinimal and \choreMaximal.
Furthermore, the interest sets of the meta-goods are pairwise disjoint; that is, $T_i \cap T_j = \emptyset$ for any pair of~$i, j \in \{1, \dots, \widehat{\ell}\}$ with $i \neq j$.

\item For all~$i \in N$, $S \subseteq \{g_1, \dots, g_s, M_1, \dots, M_{\widehat{\ell}}\}$, and $c \in \widehat{Z}$, we have $u_i(\{c\} \cup S) < 0$.
\end{itemize}

To achieve the above goals, \Cref{alg:chore-max-good-min} takes as input the set of meta-goods $\{M_1, \dots, M_\ell\}$ and the set of objective chores~$Z$ outputted by \Cref{alg:itemmerge}.
Initialize $X \gets \emptyset$, which will contain the subjective goods being removed from~$\bigcup_{j \in [\ell]} M_j$ during the execution of \Cref{alg:chore-max-good-min}, $\mathcal{M} \gets \{M_1, \dots, M_\ell\}$, which is the initial set of the input meta-goods, $\widehat{Z} \gets Z$, which is the initial set of the input objective chores, and an index~$\widehat{\ell} \gets \ell$.
At a high level, the \verb|while|-loop-condition in \cref{alg:chore-max-good-min:bundle-condition} examines if either of the following two conditions holds:
\begin{itemize}
\item there exists~$M_j \in \mathcal{M}$, $i \in N$ and~$g \in M_j$ such that $u_i(g) \geq 0$ and $u_i(M_j \setminus \{g\}) > 0$, meaning that $M_j$ violates \goodMinimality; or

\item $\widehat{Z}$ is non-empty and moreover, there exists~$c \in \widehat{Z}$ and $S \subseteq X \cup \mathcal{M}$ such that $u_i(S \cup \{c\}) \geq 0$.

Note that this condition also checks if some~$M_j$ violates \choreMaximality, by simply taking $S \gets \{M_j\}$.
\end{itemize}
In either case, we will modify~$\mathcal{M}$.
We maintain that each agent values at most one~$M_j \in \mathcal{M}$ non-negatively.
Each iteration of the \verb|while|-loop in \crefrange{alg:chore-max-good-min:merge-begin}{alg:chore-max-good-min:merge-end} merges two meta-goods~$M_r$ and~$M_{r'}$ into~$M_r$ if some agent~$i' \in N$ values both~$M_r$ and~$M_{r'}$ non-negatively, and removes~$M_{r'}$ from~$\mathcal{M}$.
Clearly, the updated~$M_r$ remains a meta-good.
Since $|\mathcal{M}|$ is polynomial in~$m$ and strictly decreases, each execution of the \verb|while|-loop in \crefrange{alg:chore-max-good-min:merge-begin}{alg:chore-max-good-min:merge-end} will terminate in polynomial time.
At its termination, each agent values at most one $M_j \in \mathcal{M}$ non-negatively.

In the case that the \verb|if|-statement is executed, \Cref{alg:chore-max-good-min} identifies an~$M_j$, an agent~$i \in N$ and her subjective good~$g \in M_j$ such that $u_i(g) \geq 0$ and $u_i(M_j \setminus \{g\}) > 0$.
We then update~$M_j \gets M_j \setminus \{g\}$ and $X \gets X \cup \{g\}$ by moving~$g$ from~$M_j$ to~$X$.
Clearly, the updated~$M_j$ remains a meta-good.
Moreover, it can be verified easily that these operations can be done in polynomial time.

In the case that the \verb|else|-statement is executed, \Cref{alg:chore-max-good-min} identifies an agent~$i \in N$, a subset~$S \subseteq X \cup \mathcal{M}$ and an objective chore~$c \in \widehat{Z}$ such that $u_i(S \cup \{c\}) \geq 0$.
We first update $\widehat{\ell} \gets \widehat{\ell} + 1$ and next create the meta-good $M_{\widehat{\ell}} \gets S \cup \{c\}$.
Clearly, $M_{\widehat{\ell}}$ is a meta-good.
Also, we update $\mathcal{M} \gets (\mathcal{M} \setminus S) \cup \{M_{\widehat{\ell}}\}$ by removing those meta-goods contained in~$S$ and adding the new meta-good~$M_{\widehat{\ell}}$ to~$\mathcal{M}$, as well as update $X \gets X \setminus S$ and $\widehat{Z} \gets \widehat{Z} \setminus \{c\}$ accordingly.
We now describe how these operations can be done in polynomial time.
For a fixed agent~$i \in N$ and chore~$c \in \widehat{Z}$, the best subset~$S \subseteq X \cup \mathcal{M}$ consists of all current meta-goods that agent~$i$ values positively.
Thus, the subset test is equivalent to
\[
u_i(c) + \sum_{e \in X \cup \mathcal{M}} \max\{u_i(e), 0\} \geq 0;
\]
when the inequality holds, those positively-valued meta-goods give a witness~$S$ directly.
Clearly, this can be done in polynomial time.
Furthermore, applying this idea over all agents~$N$ and all current chores in~$\widehat{Z}$ runs in polynomial time.

We now prove that the outer \verb|while|-loop (\crefrange{alg:chore-max-good-min:bundle-condition}{alg:chore-max-good-min:merge-end}) terminates in polynomial time.
In each iteration of \crefrange{alg:chore-max-good-min:bundle-condition}{alg:chore-max-good-min:merge-end}, either the \verb|if|- or the \verb|else|-statement is executed.
First, $|\widehat{Z}|$ strictly decreases when executing the \verb|else|-statement and remains unchanged when \crefrange{alg:chore-max-good-min:bundle-condition}{alg:chore-max-good-min:merge-end} performing all other steps.
The number of times the \verb|else|-statement is executed is upper bounded by $|Z| \leq n - 1$.
Next, between any two consecutive executions of the \verb|else|-statement, the cardinality size of $X \cup \bigcup_{M_j \in \mathcal{M}} M_j$, which is equal to $|M \setminus \widehat{Z}|$, stays the same.
Since $|X|$ strictly increases when executing the \verb|if|-statement, the number of times the \verb|if|-statement is executed between any two consecutive executions of the \verb|else|-statement is upper bounded by~$m$.
The same bound applies before the first and after the last execution of the \texttt{else}-statement, giving at most $|Z| + 1$ phases and at most $m (|Z| + 1)$ executions of the \texttt{if}-statement in total.
Together with the fact that each component inside the \verb|while|-loop (\crefrange{alg:chore-max-good-min:bundle-condition}{alg:chore-max-good-min:merge-end}) runs in polynomial time, we conclude that \Cref{alg:chore-max-good-min} runs in polynomial time.

At the termination of the outer \verb|while|-loop (\crefrange{alg:chore-max-good-min:bundle-condition}{alg:chore-max-good-min:merge-end}), if $|\widehat{Z}| = 0$, the instance can be processed by \Cref{alg:meta-goods} which was discussed in \Cref{sec:EFM:<n-chores:=0}.
Otherwise, we have $|\widehat{Z}| \geq 1$.
Since the number of objective chores only decreases in our algorithm, we have $|\widehat{Z}| \leq n-1$.
To summarize, after relabelling all meta-goods~$M_j \in \mathcal{M}$ to obtain the updated set of meta-goods $\{M_1, \dots, M_{\widehat{\ell}}\}$, we have the following lemma.

\begin{lemma}
\label{lem:matchProprty}
Under Case~II.2, if the updated chore set~$\widehat{Z}$ is non-empty, \Cref{alg:chore-max-good-min} outputs, in polynomial time, a partition of the item set~$M$ into meta-goods $\{M_1, \dots, M_{\widehat{\ell}}\}$, a set of subjective goods $X = \{g_1, \cdots, g_s\}$, and a set of objective chores $\widehat{Z} = \{c_1, \dots, c_k\}$ with $1 \leq k \leq n-1$.
These sets satisfy the following properties:
\begin{enumerate}[label=(P\arabic*)]
\item Interest sets are pairwise disjoint, i.e., $T_i \cap T_j = \emptyset$ for all $i \neq j$.

\item Each element of $\{M_1, \dots, M_{\widehat{\ell}}\}$ is both \goodMinimal and \choreMaximal.

\item\label{lem:matchProprty:P3}
For every objective chore~$c \in \widehat{Z}$, every agent~$i \in N$, and every subset $S \subseteq \{g_1, \dots, g_s, M_1, \dots, M_{\widehat{\ell}}\} $, we have $u_i(\{c\} \cup S) < 0$.
\end{enumerate}
\end{lemma}

Let $\mathcal P=\{g_1,\ldots,g_s,M_1,\ldots,M_{\widehat\ell}\}$ denote the singleton subjective goods and meta-goods supplied by \Cref{lem:matchProprty}.
Treating each $g_j$ as a singleton meta-good, every element of~$\mathcal P$ is \goodMinimal.
For a set $S\subseteq\mathcal P$, $u_i(S)$ denotes the utility of its underlying original items.
The notation $c_q\circ S$, or equivalently $\{c_q\}\cup S$, denotes the bundle consisting of the original chore~$c_q$ and these items.
By property~\ref{lem:matchProprty:P3}, every such bundle is a meta-chore.

We construct attachment sets $S_1^*,\ldots,S_k^*$ and a maximum-weight matching between the agents and the resulting meta-chores $c_1\circ S_1^*,\ldots,c_k\circ S_k^*$.
The fairness proof rests on three observations: each chore recipient values its attachments non-negatively, values every unattached meta-good negatively, and every agent values each constructed meta-chore negatively.
We then allocate the remaining meta-goods by \IMWM, as in \Cref{sec:EFM:<n-chores:=0}.

\paragraph{Constructing the attachment sets.}
For each $q\in[k]$, let $S_q\subseteq\mathcal P$ be the set of meta-goods attached to chore~$c_q$.
Initially, $S_q=\emptyset$ for every~$q$.
The attachment sets remain pairwise disjoint, and
\[
R=\mathcal P\setminus\bigcup_{q\in[k]}S_q
\]
is the set of unattached meta-goods.
Write $B_q=c_q\circ S_q$ and $\mathcal B=\{B_q \colon q\in[k]\}$ for the current meta-chores.
Consider the complete bipartite graph $H[N,\mathcal B]$ with edge weights $w(i,B_q)=u_i(B_q)$.
The index of each meta-chore remains fixed as its attachment set grows.

Every matching considered in $H$ covers all $k$ meta-chores; it leaves $n-k$ agents unmatched.
For such a matching~$F$, let $N(F)$ be its matched agents and write
\[
 \val_w(F)=\sum_{(i,B_q)\in F}w(i,B_q),\qquad
 \OPT(w)=\max_F\val_w(F).
\]
For each agent~$i$, define
\[
 R_i^+=\{h\in R:u_i(h)\geq0\},\qquad
 v_i=u_i(R_i^+)=\sum_{h\in R}\max\{u_i(h),0\}.
\]
Thus $R_i^+$ is the neighbourhood of~$i$ in the nonnegative-edge graph $G[N,R]$.
At each iteration, jointly choose a matching~$F$ and one of its edges $(r,B_q)$ to maximize
\begin{equation}\label{eq:expansion-score}
 \val_w(F)+v_r,
 \qquad\text{breaking ties in favour of larger }|R_r^+|.
\end{equation}
Attach $R_r^+$ to $c_q$ and repeat, or terminate if $R_r^+=\emptyset$.
For later use, denote the optimal value of this choice by
\[
 V(w,v)=\max_{F,\,i\in N(F)}\bigl(\val_w(F)+v_i\bigr).
\]

\begin{algorithm}[t]
\caption{Greedy Expansion of Meta-Chores}
\label{alg:chore-expansion}
\DontPrintSemicolon
\KwIn{Agents~$N$ with additive utilities $(u_i)_{i\in N}$, objective chores $\widehat Z=\{c_1,\ldots,c_k\}$ with $1\leq k<n$, and meta-goods $\mathcal P$ supplied by \Cref{lem:matchProprty}.}
\KwOut{Attachment sets $S_1^*,\ldots,S_k^*$ and a maximum-weight matching $F^*$ covering the constructed meta-chores.}
$S_q\gets\emptyset$ for every $q\in[k]$; $R\gets\mathcal P$\;
\While{\textup{true}}{
 $B_q\gets c_q\circ S_q$ for every $q\in[k]$;
 $\mathcal B\gets\{B_q:q\in[k]\}$\;
 $w(i,B_q)\gets u_i(B_q)$ for every $i\in N$ and $q\in[k]$\;
 \ForEach{$i\in N$}{
  $R_i^+\gets\{h\in R:u_i(h)\geq0\}$;
  $v_i\gets\sum_{h\in R_i^+}u_i(h)$\;
 }
 Compute an integral maximum-profit flow of value $k$ in \Cref{fig:expansion-flow}, with secondary profit $|R_i^+|$ on each arc $s_1\to i^-$\;
 $F\gets\{(i,B_q):\text{arc }i^+\to B_q\text{ carries one unit}\}$\;
 Let $r$ be the agent reached through $s_1$ and $B_q$ its matched meta-chore in $F$\;
 \If{$R_r^+=\emptyset$}{
  \Return{$(S_1,\ldots,S_k,F)$}\;
 }
 $S_q\gets S_q\cup R_r^+$; $R\gets R\setminus R_r^+$\;
}
\end{algorithm}

The choice in \Cref{eq:expansion-score} maximizes welfare after one expansion.
Indeed, any set $D\subseteq R$ attached to a meta-chore assigned to~$i$ adds at most $v_i$ to that matching's weight.
Attaching $R_r^+$ to the selected meta-chore attains the chosen score.
The attachment sets stay with their chores, while the matching is recomputed at each iteration.

After an expansion, every remaining meta-good is negative to the selected agent.
Since $R$ only shrinks, each agent can initiate at most one expansion.
Hence there are at most $n$ expansions and one final iteration.
Including zero-valued meta-goods in $R_i^+$ and breaking ties by $|R_i^+|$ ensures that, at termination, every chore recipient values every unattached meta-good strictly negatively.
This is established as Observation~2 in the proof of \Cref{prop:EF1+EFable:<n-chores:>0} below.

\paragraph{Polynomial implementation.}
Each choice is one integral maximum-profit flow of prescribed value $k$ in the network of \Cref{fig:expansion-flow}.
Its vertices are the source $s$, the two entry nodes $s_0,s_1$, two nodes $i^-,i^+$ for each agent $i\in N$, one node $B_q$ for each $q\in[k]$, and the sink $z$.
The figure lists every arc family; there are no other arcs.

\begin{figure}[t]
\centering
\begin{tikzpicture}[
 x=1.08cm,y=1cm,
 >=Latex,
 vertex/.style={circle,draw,line width=.45pt,minimum size=8mm,
   inner sep=1pt,font=\small,fill=white},
 bundle/.style={vertex,minimum size=10mm},
 arc/.style={->,line width=.55pt},
 entry/.style={arc,draw=black!55},
 marked/.style={arc,draw=blue!65!black,line width=.75pt},
 family/.style={->,draw=black!55,line width=.45pt,densely dashed},
 heading/.style={font=\footnotesize\bfseries,align=center}
]
\node[heading] at (0,2.15) {source};
\node[heading] at (1.7,2.15) {entry};
\node[heading] at (4,2.15) {agent input};
\node[heading] at (6,2.15) {agent output};
\node[heading] at (8.8,2.15) {meta-chores};
\node[heading] at (10.8,2.15) {sink};

\node[vertex] (s) at (0,0) {$s$};
\node[vertex] (s0) at (1.7,1.05) {$s_0$};
\node[vertex] (s1) at (1.7,-1.05) {$s_1$};
\node[vertex] (a1m) at (4,1.4) {$i_1^-$};
\node[vertex] (a2m) at (4,.1) {$i_2^-$};
\node at (4,-.65) {$\vdots$};
\node[vertex] (anm) at (4,-1.4) {$i_n^-$};
\node[vertex] (a1p) at (6,1.4) {$i_1^+$};
\node[vertex] (a2p) at (6,.1) {$i_2^+$};
\node at (6,-.65) {$\vdots$};
\node[vertex] (anp) at (6,-1.4) {$i_n^+$};
\node[bundle] (b1) at (8.8,1.4) {$B_1$};
\node at (8.8,0) {$\vdots$};
\node[bundle] (bk) at (8.8,-1.4) {$B_k$};
\node[vertex] (z) at (10.8,0) {$z$};

\draw[arc] (s) -- node[above,sloped,font=\scriptsize] {$k-1$} (s0);
\draw[marked] (s) -- node[below,sloped,font=\scriptsize] {$1$} (s1);
\foreach \a in {a1m,a2m,anm} {
 \draw[entry] (s0) -- (\a);
 \draw[marked] (s1) -- (\a);
}
\draw[arc] (a1m) -- (a1p);
\draw[arc] (a2m) -- (a2p);
\draw[arc] (anm) -- (anp);
\foreach \a in {a1p,a2p,anp} {
 \foreach \b in {b1,bk} {
  \draw[family] (\a) -- (\b);
 }
}
\draw[arc] (b1) -- (z);
\draw[arc] (bk) -- (z);
\end{tikzpicture}

\medskip
{\renewcommand{\arraystretch}{1.12}
\begin{tabular}{@{}lccc@{}}
\toprule
Arc & Indices & Capacity & Profit\\
\midrule
$s\to s_0$ & --- & $k-1$ & $0$\\
$s\to s_1$ & --- & $1$ & $0$\\
$s_0\to i^-$ & $i\in N$ & $1$ & $0$\\
$s_1\to i^-$ & $i\in N$ & $1$ & $v_i$\\
$i^-\to i^+$ & $i\in N$ & $1$ & $0$\\
$i^+\to B_q$ & $i\in N,\ q\in[k]$ & $1$ & $u_i(B_q)$\\
$B_q\to z$ & $q\in[k]$ & $1$ & $0$\\
\bottomrule
\end{tabular}}
\caption{The full network for one expansion, with $B_q=c_q\circ S_q$ for $q\in[k]$.
Both entry nodes connect to every agent input, each input connects only to its own output, and every output connects to every meta-chore.
Ellipses indicate the remaining agents and meta-chores.
The two source-arc labels show capacities; the table specifies all capacities and profits.}
\label{fig:expansion-flow}
\end{figure}

The prescribed flow value is $k$, regardless of the signs of the profits.
The source capacities force exactly one unit through $s_1$; its agent is the selected recipient $r$.
The $k$ unit-capacity arcs into $z$ force every meta-chore to be assigned, and $i^-\to i^+$ ensures distinct recipients.
Conversely, any matching $F$ and selected recipient $r\in N(F)$ define such a flow, routing $r$ through $s_1$ and the other recipients through $s_0$.
Its profit is exactly $\val_w(F)+v_r$.
When $k=1$, the ordinary entry simply has capacity zero.
The secondary profit $|R_i^+|$ on $s_1\to i^-$ implements the tie-break in \Cref{eq:expansion-score}; any remaining ties may be broken arbitrarily.
There are $O(n)$ vertices and $O(n^2)$ arcs, and at most $n+1$ flow computations.
For rational utilities, clear denominators, multiply primary integer profits by $m+1$, and add the secondary profit; this implements the tie-break with polynomial bit length.

Although each expansion attaches meta-goods valued non-negatively by the selected agent, later iterations may assign the expanded meta-chore to another agent.
The following lemma ensures that its final recipient still values the attachment set non-negatively.

\begin{lemma}[Non-negative attachment utility]\label{lem:expansion-attachment-safety}
At every stage of \Cref{alg:chore-expansion}, every maximum-weight matching assigns each meta-chore $B_q=c_q\circ S_q$ to an agent~$i$ with $u_i(S_q)\geq0$.
\end{lemma}

The proof in \Cref{app:expansion-attachment-proof} uses an alternating-path bound to maintain a stronger invariant: resetting any one expanded meta-chore to its original chore, even followed by the best next expansion, cannot improve on the current maximum matching weight.
We apply this lemma in Observation~1 below.

\paragraph{The resulting allocation.}
Let $(S_1^*,\ldots,S_k^*,F^*)$ be the output of \Cref{alg:chore-expansion}.
Define
\[
 N^d=N\setminus N(F^*),\qquad
 R^*=\mathcal P\setminus\bigcup_{q\in[k]}S_q^*.
\]
Relabel the agents in $N\setminus N^d$ so that agent~$i$ is matched to $c_i\circ S_i^*$, for each $i\in[k]$, and label the remaining agents so that $N^d=\{k+1,\ldots,n\}$.
Allocate $A_i=\{c_i\}\cup S_i^*$ to each $i\in[k]$.
Allocate the remaining meta-goods $R^*$ to the agents in $N^d$ using \IMWM in $G[N^d,R^*]$, as in Case~II.1.
In each round, choose a maximum-cardinality matching among maximum-weight matchings.

Introduce $n-k$ dummy chores $d_1,\ldots,d_{n-k}$, each valued at zero by every agent, and define
\[
 \widetilde J(S_1^*,\ldots,S_k^*)
 =\{d_1,\ldots,d_{n-k}\}\cup
   \{c_i\circ S_i^*:i\in[k]\}.
\]
Assigning these dummy chores to the agents in $N^d$ extends $F^*$ to a perfect matching in $H[N,\widetilde J(S_1^*,\ldots,S_k^*)]$.

\begin{proposition}
\label{prop:EF1+EFable:<n-chores:>0}
Under Case~II.2, the allocation $A=(A_1,\ldots,A_n)$ described above is complete, \EFone, and \EFable.
The allocation can be computed in polynomial time.
\end{proposition}

\begin{proof}
We first establish three observations used in the fairness proof.

\medskip
\noindent\textit{Observation~1: Each agent~$i\in N\setminus N^d$ derives non-negative total utility from its attachment set, i.e., $u_i(S_i^*)\geq0$.}
\Cref{lem:expansion-attachment-safety} shows that every maximum-weight matching assigns each meta-chore to an agent who values its attachment set non-negatively.
Applying it to~$F^*$ proves the observation.
This concerns the total utility of $S_i^*$; individual meta-goods in that set may have negative utility for~$i$.

\medskip
\noindent\textit{Observation~2: Each agent~$i\in N\setminus N^d$ derives negative utility from every remaining meta-good in $\mathcal P\setminus\bigcup_{q\in[k]}S_q^*$.}
Suppose that $R_i^+\ne\emptyset$ for some matched agent~$i$ at termination.
If $v_i^*>0$, selecting~$i$ with the same matching~$F^*$ would increase the expansion score.
If $v_i^*=0$, it would leave the score unchanged and improve the tie-break, since the selected agent~$r$ has $|R_r^+|=0$.
Both cases contradict the algorithm's choice.
Hence $R_i^+=\emptyset$ for every chore recipient~$i$.

\medskip
\noindent\textit{Observation~3: For each agent~$r\in N$ and any $i\in[k]$, we have $u_r(A_i)=u_r(\{c_i\}\cup S_i^*)<0$.}
This follows from property~\ref{lem:matchProprty:P3} of \Cref{lem:matchProprty}, since $c_i\in\widehat Z$ and $S_i^*\subseteq\mathcal P$.

\medskip
Each remaining meta-good has a non-negative valuer in~$N$.
By Observation~2, this valuer belongs to~$N^d$.
Moreover, \goodMinimality is preserved when restricting the agents to~$N^d$.
Thus \Cref{alg:meta-goods} applies to $G[N^d,R^*]$ without further splitting.
Since $\mathcal P$ and $\widehat Z$ partition $M$, and the attachment sets are pairwise disjoint, the allocation~$A$ is complete.

\medskip
\noindent\textbf{Allocation~$A$ is \EFone.}
Fix an agent $i\in N\setminus N^d$.
By Observation~1, removing the single original chore~$c_i$ leaves non-negative utility.
Every other chore bundle has negative utility for~$i$ by Observation~3.
Each bundle allocated to an agent in~$N^d$ consists of remaining meta-goods, so it has non-positive utility for~$i$ by Observation~2.
Consequently, for every $j\in N$,
\[
 u_i(A_i\setminus\{c_i\})=u_i(S_i^*)\geq0\geq u_i(A_j).
\]
Hence every agent in $N\setminus N^d$ satisfies \EFone.

The agents in~$N^d$ satisfy \EFone among themselves by \Cref{prop:EF1+EFable:<n-chores:=0}.
Each receives non-negative utility, whereas every chore bundle has negative utility for that agent by Observation~3.
Thus no agent in~$N^d$ envies a chore recipient, proving that~$A$ is \EFone.

\medskip
\noindent\textbf{The output matching is maximum-weight.}
We now establish the maximum-weight matching property.
Let $w^*$ and $v^*$ be the weights and residual utilities at termination.
The selected agent~$r$ has $R_r^+=\emptyset$, and hence $v_r^*=0$.
Since $F^*$ and $r$ attain the best expansion score,
\[
 V(w^*,v^*)=\val_{w^*}(F^*)\leq \OPT(w^*)\leq V(w^*,v^*).
\]
The last inequality holds because all residual utilities are non-negative.
Thus $\val_{w^*}(F^*)=\OPT(w^*)$: the output matching maximizes total utility for the fixed constructed meta-chores.
Since the dummy chores have zero utility, its extension is a maximum-weight perfect matching in $H[N,\widetilde J(S_1^*,\ldots,S_k^*)]$.

\medskip
\noindent\textbf{Allocation~$A$ is \EFable.}
Decompose~$A$ into
\[
 Q=(A_1,\ldots,A_k,\emptyset,\ldots,\emptyset),\qquad
 W=(\emptyset,\ldots,\emptyset,A_{k+1},\ldots,A_n).
\]
The maximum-weight perfect matching property proved above implies that~$Q$ maximizes welfare among permutations of its bundles.
By Observation~2, the agents in $N\setminus N^d$ are isolated in $G[N,R^*]$ and in every subsequent residual graph.
Hence the same sequence of maximum-weight matchings used on $N^d$ can also be used on all of~$N$, giving allocation~$W$.
By \Cref{prop:EF1+EFable:<n-chores:=0} and \Cref{thm:hs_ef}, $W$ also maximizes welfare among permutations of its bundles.
For every permutation $\sigma$ of~$N$, additivity therefore gives
\[
\sum_{i \in N} u_i(A_i) = \sum_{i \in N} u_i(Q_i) + \sum_{i \in N} u_i(W_i) \geq \sum_{i \in N} u_i(Q_{\sigma(i)}) +\sum_{i \in N} u_i(W_{\sigma(i)}) = \sum_{i \in N} u_i(A_{\sigma(i)}).
\]
By \Cref{thm:hs_ef}, $A$ is \EFable.

\medskip
\noindent\textbf{Polynomial-time computation.}
The refinement is polynomial by \Cref{lem:matchProprty}, and \Cref{alg:chore-expansion} uses at most $n+1$ flow computations of polynomial size.
In the final \IMWM, every remaining meta-good has an incident edge in the nonnegative-edge graph.
The maximum-cardinality tie-break therefore allocates at least one meta-good in every nonempty round, giving at most $|\mathcal P|\leq m$ rounds.
All weights are sums of input utilities, and exact comparisons and tie-breaking have polynomial bit complexity.
\end{proof}

Combining \Cref{prop:EF1+EFable:>=n-chores,prop:EF1+EFable:<n-chores:=0,prop:EF1+EFable:<n-chores:>0}, we establish \Cref{thm:EF1+EFable-existence}.

\section{Discussion}

In this paper, we have studied the general fair division problem that involves indivisible goods and chores as well as a cake, and show the existence of \EFM allocations.
To establish this result, we prove that with indivisible goods and chores, there always exists an allocation that is both \EFone and \EFable.
Moreover, we devise polynomial-time algorithms to compute such an allocation.
Our findings uncover an interesting connection between solution concepts of envy-freeness, \EFability, \EFone, and \EFM that have been investigated in closely related but separately developed lines of research.

A potential direction for future work is to consider the more general class of valuations beyond additive.
For instance, with only indivisible items, \citet{BhaskarSrVa21} devised an algorithm that always finds an \EFone allocation for doubly monotone instances (cf.\ \Cref{ft:doubly-monotone}).
Our iterative matching techniques can no longer extend beyond additive valuations.
This raises the following intriguing questions.
Does an \EFM allocation always exist for agents with doubly monotone valuations over indivisible items?

\section*{Acknowledgements}

This work was partially supported
by the NSF-CSIRO grant on ``Fair Sequential Collective Decision Making'',
by the ARC Laureate Project FL200100204 on ``Trustworthy AI'', and
by JST ERATO Grant Number JPMJER2301.

We would like to thank Shengxin Liu and the anonymous reviewers for their valuable feedback.

Generative AI tools, especially ChatGPT (GPT-6 Astra, accessed October 2026), assisted in writing the polynomial-time arguments.

\bibliographystyle{plainnat}
\bibliography{bibliography}

\begin{thebibliography}{52}
\providecommand{\natexlab}[1]{#1}
\providecommand{\url}[1]{\texttt{#1}}
\expandafter\ifx\csname urlstyle\endcsname\relax
  \providecommand{\doi}[1]{doi: #1}\else
  \providecommand{\doi}{doi: \begingroup \urlstyle{rm}\Url}\fi

\bibitem[Alon(1987)]{Alon87}
Noga Alon.
\newblock Splitting necklaces.
\newblock \emph{Advances in Mathematics}, 63\penalty0 (3):\penalty0 247--253,
  1987.

\bibitem[Amanatidis et~al.(2023)Amanatidis, Aziz, Birmpas, Filos-Ratsikas, Li,
  Moulin, Voudouris, and Wu]{AmanatidisAzBi23}
Georgios Amanatidis, Haris Aziz, Georgios Birmpas, Aris Filos-Ratsikas, Bo~Li,
  Herv\'{e} Moulin, Alexandros~A. Voudouris, and Xiaowei Wu.
\newblock Fair division of indivisible goods: Recent progress and open
  questions.
\newblock \emph{Artificial Intelligence}, 322:\penalty0 103965, 2023.

\bibitem[Aragones(1995)]{Arag95a}
Enriqueta Aragones.
\newblock A derivation of the money {R}awlsian solution.
\newblock \emph{Social Choice and Welfare}, 12\penalty0 (3):\penalty0 267--276,
  1995.

\bibitem[Aziz and Mackenzie(2016{\natexlab{a}})]{AzizMa16}
Haris Aziz and Simon Mackenzie.
\newblock A discrete and bounded envy-free cake cutting protocol for any number
  of agents.
\newblock In \emph{Proceedings of the 57th Annual IEEE Symposium on Foundations
  of Computer Science (FOCS)}, pages 416--427, 2016{\natexlab{a}}.

\bibitem[Aziz and Mackenzie(2016{\natexlab{b}})]{AzizMa16-STOC}
Haris Aziz and Simon Mackenzie.
\newblock A discrete and bounded envy-free cake cutting protocol for four
  agents.
\newblock In \emph{Proceedings of the 48th Annual ACM Symposium on Theory of
  Computing (STOC)}, pages 454--464, 2016{\natexlab{b}}.

\bibitem[Aziz et~al.(2020)Aziz, Moulin, and Sandomirskiy]{AMS20a}
Haris Aziz, Herv\'{e} Moulin, and Fedor Sandomirskiy.
\newblock A polynomial-time algorithm for computing a {P}areto optimal and
  almost proportional allocation.
\newblock \emph{Operations Research Letters}, 48\penalty0 (5):\penalty0
  573--578, 2020.

\bibitem[Aziz et~al.(2022)Aziz, Caragiannis, Igarashi, and Walsh]{AzizCaIg22}
Haris Aziz, Ioannis Caragiannis, Ayumi Igarashi, and Toby Walsh.
\newblock Fair allocation of indivisible goods and chores.
\newblock \emph{Autonomous Agents and Multi-Agent Systems}, 36\penalty0
  (1):\penalty0 3:1--3:21, 2022.

\bibitem[Aziz et~al.(2025)Aziz, He, Lu, and Zhou]{AzizHeLu25}
Haris Aziz, Zixu He, Xinhang Lu, and Kaiyang Zhou.
\newblock Fair allocation of divisible goods under non-linear valuations.
\newblock In \emph{Proceedings of the 24th International Conference on
  Autonomous Agents and Multiagent Systems (AAMAS)}, pages 170--178, 2025.

\bibitem[Aziz et~al.(2027)Aziz, Bu, Lu, Mackenzie, Suzuki, Tao, and
  Walsh]{AzizBuLu27}
Haris Aziz, Xiaolin Bu, Xinhang Lu, Simon Mackenzie, Mashbat Suzuki, Biaoshuai
  Tao, and Toby Walsh.
\newblock Best-of-both-worlds fairness for mixed goods and chores.
\newblock In \emph{Proceedings of the 38th ACM-SIAM Symposium on Discrete
  Algorithms (SODA)}, 2027.
\newblock Forthcoming.

\bibitem[Barman and Verma(2026)]{BarmanVe26}
Siddharth Barman and Paritosh Verma.
\newblock Introspectively envy-free and efficient allocation of indivisible
  mixed manna.
\newblock \emph{CoRR}, abs/2509.18673, 2026.

\bibitem[Barman et~al.(2025)Barman, HV, Sethia, and Suzuki]{BarmanHVSe25}
Siddharth Barman, Vishwa~Prakash HV, Aditi Sethia, and Mashbat Suzuki.
\newblock Fair and efficient allocation of indivisible mixed manna.
\newblock In \emph{Proceedings of the 21th Conference on Web and Internet
  Economics (WINE)}, pages 467--483, 2025.

\bibitem[Bei et~al.(2021{\natexlab{a}})Bei, Li, Liu, Liu, and Lu]{BeiLiLi21}
Xiaohui Bei, Zihao Li, Jinyan Liu, Shengxin Liu, and Xinhang Lu.
\newblock Fair division of mixed divisible and indivisible goods.
\newblock \emph{Artificial Intelligence}, 293:\penalty0 103436,
  2021{\natexlab{a}}.

\bibitem[Bei et~al.(2021{\natexlab{b}})Bei, Liu, Lu, and Wang]{BeiLiLu21}
Xiaohui Bei, Shengxin Liu, Xinhang Lu, and Hongao Wang.
\newblock Maximin fairness with mixed divisible and indivisible goods.
\newblock \emph{Autonomous Agents and Multi-Agent Systems}, 35\penalty0
  (2):\penalty0 34:1--34:21, 2021{\natexlab{b}}.

\bibitem[Bei et~al.(2025)Bei, Liu, and Lu]{BeiLiLu25}
Xiaohui Bei, Shengxin Liu, and Xinhang Lu.
\newblock Fair division with subjective divisibility.
\newblock \emph{Games and Economic Behavior}, 151:\penalty0 127--147, 2025.

\bibitem[Benad\`{e} et~al.(2024)Benad\`{e}, Kazachkov, Procaccia, Psomas, and
  Zeng]{BenadeKaPr24}
Gerdus Benad\`{e}, Aleksandr~M. Kazachkov, Ariel~D. Procaccia, Alexandros
  Psomas, and David Zeng.
\newblock Fair and efficient online allocations.
\newblock \emph{Operations Research}, 72\penalty0 (4):\penalty0 1438--1452,
  2024.

\bibitem[Bhaskar et~al.(2021)Bhaskar, Sricharan, and Vaish]{BhaskarSrVa21}
Umang Bhaskar, A.~R. Sricharan, and Rohit Vaish.
\newblock On approximate envy-freeness for indivisible chores and mixed
  resources.
\newblock In \emph{Proceedings of the 24th International Conference on
  Approximation Algorithms for Combinatorial Optimization Problems (APPROX)},
  pages 1:1--1:23, 2021.

\bibitem[Bogomolnaia et~al.(2017)Bogomolnaia, Moulin, Sandomirskiy, and
  Yanovskaya]{BogomolnaiaMoSa17}
Anna Bogomolnaia, Herv\'{e} Moulin, Fedor Sandomirskiy, and Elena Yanovskaya.
\newblock Competitive division of a mixed manna.
\newblock \emph{Econometrica}, 85\penalty0 (6):\penalty0 1847--1871, 2017.

\bibitem[Brams and Taylor(1996)]{BramsTa96}
Steven~J. Brams and Alan~D. Taylor.
\newblock \emph{Fair Division: From Cake-Cutting to Dispute Resolution}.
\newblock Cambridge University Press, 1996.

\bibitem[Brustle et~al.(2020)Brustle, Dippel, Narayan, Suzuki, and
  Vetta]{BrustleDiNa20}
Johannes Brustle, Jack Dippel, Vishnu~V. Narayan, Mashbat Suzuki, and Adrian
  Vetta.
\newblock One dollar each eliminates envy.
\newblock In \emph{Proceedings of the 21st ACM Conference on Economics and
  Computation (EC)}, pages 23--39, 2020.

\bibitem[Bu et~al.(2024)Bu, Li, Liu, Lu, and Tao]{BuLiLi24}
Xiaolin Bu, Zihao Li, Shengxin Liu, Xinhang Lu, and Biaoshuai Tao.
\newblock Best-of-both-worlds fair allocation of indivisible and mixed goods.
\newblock In \emph{Proceedings of the 20th Conference on Web and Internet
  Economics (WINE)}, pages 277--294, 2024.

\bibitem[Budish(2011)]{Budish11}
Eric Budish.
\newblock The combinatorial assignment problem: Approximate competitive
  equilibrium from equal incomes.
\newblock \emph{Journal of Political Economy}, 119\penalty0 (6):\penalty0
  1061--1103, 2011.

\bibitem[Caragiannis and Ioannidis(2021)]{CaragiannisIo21}
Ioannis Caragiannis and Stavros Ioannidis.
\newblock Computing envy-freeable allocations with limited subsidies.
\newblock In \emph{Proceedings of the 17th International Conference on Web and
  Internet Economics (WINE)}, pages 522--539, 2021.

\bibitem[Caragiannis et~al.(2019)Caragiannis, Kurokawa, Moulin, Procaccia,
  Shah, and Wang]{CaragiannisKuMo19}
Ioannis Caragiannis, David Kurokawa, Herv\'{e} Moulin, Ariel~D. Procaccia,
  Nisarg Shah, and Junxing Wang.
\newblock The unreasonable fairness of maximum {N}ash welfare.
\newblock \emph{ACM Transactions on Economics and Computation}, 7\penalty0
  (3):\penalty0 12:1--12:32, 2019.

\bibitem[Chaudhury et~al.(2023)Chaudhury, Garg, McGlaughlin, and
  Mehta]{ChaudhuryGaMc23}
Bhaskar~Ray Chaudhury, Jugal Garg, Peter McGlaughlin, and Ruta Mehta.
\newblock A complementary pivot algorithm for competitive allocation of a mixed
  manna.
\newblock \emph{Mathematics of Operations Research}, 48\penalty0 (3):\penalty0
  1630--1656, 2023.

\bibitem[Gal et~al.(2017)Gal, Mash, Procaccia, and Zick]{GMPZ17a}
Ya'akov~(Kobi) Gal, Moshe Mash, Ariel~D. Procaccia, and Yair Zick.
\newblock Which is the fairest (rent division) of them all?
\newblock \emph{Journal of the ACM (JACM)}, 64\penalty0 (6):\penalty0
  39:1--39:22, 2017.

\bibitem[Goko et~al.(2024)Goko, Igarashi, Kawase, Makino, Sumita, Tamura,
  Yokoi, and Yokoo]{GokoIgKa24}
Hiromichi Goko, Ayumi Igarashi, Yasushi Kawase, Kazuhisa Makino, Hanna Sumita,
  Akihisa Tamura, Yu~Yokoi, and Makoto Yokoo.
\newblock A fair and truthful mechanism with limited subsidy.
\newblock \emph{Games and Economic Behavior}, 144:\penalty0 49--70, 2024.

\bibitem[Guo et~al.(2023)Guo, Li, and Deng]{GuoLiDe23}
Hao Guo, Weidong Li, and Bin Deng.
\newblock A survey on fair allocation of chores.
\newblock \emph{Mathematics}, 11\penalty0 (16):\penalty0 3616, 2023.

\bibitem[Haake et~al.(2002)Haake, Raith, and Su]{HRS02a}
Claus-Jochen Haake, Matthias~G. Raith, and Francis~Edward Su.
\newblock Bidding for envy-freeness: A procedural approach to $n$-player
  fair-division problems.
\newblock \emph{Social Choice and Welfare}, 19\penalty0 (4):\penalty0
  723---749, 2002.

\bibitem[Halpern and Shah(2019)]{HalpernSh19}
Daniel Halpern and Nisarg Shah.
\newblock Fair division with subsidy.
\newblock In \emph{Proceedings of the 12th International Symposium on
  Algorithmic Game Theory (SAGT)}, pages 374--389, 2019.

\bibitem[He and Tao(2027)]{HeTa27}
Wentao He and Biaoshuai Tao.
\newblock {EFX} for additive chores: Nonexistence, {P}areto incompatibility,
  and bi-valued existence.
\newblock In \emph{Proceedings of the 38th ACM-SIAM Symposium on Discrete
  Algorithms (SODA)}, 2027.
\newblock Forthcoming.

\bibitem[Hosseini et~al.(2023)Hosseini, Sikdar, Vaish, and Xia]{HosseiniSiVa23}
Hadi Hosseini, Sujoy Sikdar, Rohit Vaish, and Lirong Xia.
\newblock Fairly dividing mixtures of goods and chores under lexicographic
  preferences.
\newblock In \emph{Proceedings of the 22nd International Conference on
  Autonomous Agents and Multiagent Systems (AAMAS)}, pages 152--160, 2023.

\bibitem[Kawase et~al.(2025)Kawase, Makino, Sumita, Tamura, and
  Yokoo]{KawaseMaSu25}
Yasushi Kawase, Kazuhisa Makino, Hanna Sumita, Akihisa Tamura, and Makoto
  Yokoo.
\newblock Towards optimal subsidy bounds for envy-freeable allocations.
\newblock \emph{Artificial Intelligence}, 348:\penalty0 104406, 2025.

\bibitem[Klijn(2000)]{Klij00a}
Flip Klijn.
\newblock An algorithm for envy-free allocations in an economy with indivisible
  objects and money.
\newblock \emph{Social Choice and Welfare}, 17\penalty0 (2):\penalty0 201--215,
  2000.

\bibitem[Kulkarni et~al.(2021{\natexlab{a}})Kulkarni, Mehta, and Taki]{KMT21a}
Rucha Kulkarni, Ruta Mehta, and Setareh Taki.
\newblock On the {PTAS} for maximin shares in an indivisible mixed manna.
\newblock In \emph{Proceedings of the 35th AAAI Conference on Artificial
  Intelligence (AAAI)}, pages 5523--5530, 2021{\natexlab{a}}.

\bibitem[Kulkarni et~al.(2021{\natexlab{b}})Kulkarni, Mehta, and Taki]{KMT21b}
Rucha Kulkarni, Ruta Mehta, and Setareh Taki.
\newblock Indivisible mixed manna: On the computability of {MMS} + {PO}
  allocations.
\newblock In \emph{Proceedings of the 22nd ACM Conference on Economics and
  Computation (EC)}, pages 683--684, 2021{\natexlab{b}}.

\bibitem[Li et~al.(2024{\natexlab{a}})Li, Li, Liu, and Wu]{LiLiLi24}
Bo~Li, Zihao Li, Shengxin Liu, and Zekai Wu.
\newblock Allocating mixed goods with customized fairness and indivisibility
  ratio.
\newblock In \emph{Proceedings of the 33rd International Joint Conference on
  Artificial Intelligence (IJCAI)}, pages 2868--2876, 2024{\natexlab{a}}.

\bibitem[Li et~al.(2023)Li, Liu, Lu, and Tao]{LiLiLu23}
Zihao Li, Shengxin Liu, Xinhang Lu, and Biaoshuai Tao.
\newblock Truthful fair mechanisms for allocating mixed divisible and
  indivisible goods.
\newblock In \emph{Proceedings of the 32nd International Joint Conference on
  Artificial Intelligence (IJCAI)}, pages 2808--2816, 2023.

\bibitem[Li et~al.(2024{\natexlab{b}})Li, Liu, Lu, Tao, and Tao]{LiLiLu24}
Zihao Li, Shengxin Liu, Xinhang Lu, Biaoshuai Tao, and Yichen Tao.
\newblock A complete landscape for the price of envy-freeness.
\newblock In \emph{Proceedings of the 23rd International Conference on
  Autonomous Agents and Multiagent Systems (AAMAS)}, pages 1183--1191,
  2024{\natexlab{b}}.

\bibitem[Lipton et~al.(2004)Lipton, Markakis, Mossel, and Saberi]{LiptonMaMo04}
Richard~J. Lipton, Evangelos Markakis, Elchanan Mossel, and Amin Saberi.
\newblock On approximately fair allocations of indivisible goods.
\newblock In \emph{Proceedings of the 5th ACM Conference on Electronic Commerce
  (EC)}, pages 125--131, 2004.

\bibitem[Liu et~al.(2024)Liu, Lu, Suzuki, and Walsh]{LiuLuSu24}
Shengxin Liu, Xinhang Lu, Mashbat Suzuki, and Toby Walsh.
\newblock Mixed fair division: {A} survey.
\newblock \emph{Journal of Artificial Intelligence Research}, 80:\penalty0
  1373--1406, 2024.

\bibitem[Lu et~al.(2026)Lu, Mackenzie, and Suzuki]{LuMaSu26}
Xinhang Lu, Simon Mackenzie, and Mashbat Suzuki.
\newblock Optimal subsidy bounds for goods and chores: One dollar each
  suffices.
\newblock \emph{CoRR}, abs/2607.10089, 2026.

\bibitem[Mishra et~al.(2023)Mishra, Padala, and Gujar]{MPS23a}
Shaily Mishra, Manisha Padala, and Sujit Gujar.
\newblock Fair allocation of goods and chores -- tutorial and survey of recent
  results.
\newblock \emph{CoRR}, abs/2307.10985, 2023.

\bibitem[Nguyen and Rothe(2023)]{NguyenRo23}
Trung~Thanh Nguyen and J\"{o}rg Rothe.
\newblock Complexity results and exact algorithms for fair division of
  indivisible items: {A} survey.
\newblock In \emph{Proceedings of the 32nd International Joint Conference on
  Artificial Intelligence (IJCAI)}, pages 6732--6740, 2023.

\bibitem[Nishimura and Sumita(2025)]{NishimuraSu25}
Koichi Nishimura and Hanna Sumita.
\newblock Envy-freeness and maximum {N}ash welfare for mixed divisible and
  indivisible goods.
\newblock \emph{Mathematical Social Sciences}, 138:\penalty0 102449, 2025.

\bibitem[Procaccia(2013)]{Procaccia13}
Ariel~D. Procaccia.
\newblock Cake cutting: Not just child's play.
\newblock \emph{Communications of the ACM}, 56\penalty0 (7):\penalty0 78--87,
  2013.

\bibitem[Robertson and Webb(1998)]{RobertsonWe98}
Jack Robertson and William Webb.
\newblock \emph{Cake-Cutting Algorithm: Be Fair If You Can}.
\newblock A K Peters/CRC Press, 1998.

\bibitem[Steinhaus(1949)]{Steinhaus49}
Hugo Steinhaus.
\newblock Sur la division pragmatique.
\newblock \emph{Econometrica}, 17:\penalty0 315--319, 1949.

\bibitem[Stromquist and Woodall(1985)]{StroWoo85}
Walter Stromquist and Douglas~R Woodall.
\newblock Sets on which several measures agree.
\newblock \emph{Journal of Mathematical Analysis and Applications},
  108\penalty0 (1):\penalty0 241--248, 1985.

\bibitem[Su(1999)]{Su99}
Francis~Edward Su.
\newblock Rental harmony: {S}perner's lemma in fair division.
\newblock \emph{The American Mathematical Monthly}, 106\penalty0 (10):\penalty0
  930--942, 1999.

\bibitem[Suksompong(2021)]{Suksompong21}
Warut Suksompong.
\newblock Constraints in fair division.
\newblock \emph{ACM SIGecom Exchanges}, 19\penalty0 (2):\penalty0 46--61, 2021.

\bibitem[Suksompong(2025)]{Suksompong25}
Warut Suksompong.
\newblock Weighted fair division of indivisible items: {A} review.
\newblock \emph{Information Processing Letters}, 187:\penalty0 106519, 2025.

\bibitem[Svensson(1983)]{Sven83a}
Lars-Gunnar Svensson.
\newblock Large indivisibles: An analysis with respect to price equilibrium and
  fairness.
\newblock \emph{Econometrica}, 51\penalty0 (4):\penalty0 939--954, 1983.

\end{thebibliography}

\appendix

\section{Attachment Values Under Greedy Expansion}
\label{app:expansion-attachment-proof}

We prove \Cref{lem:expansion-attachment-safety} using an alternating-path
bound and a reset invariant. Throughout this appendix, a matching covers
all $k$ real meta-chores in $H[N,\mathcal B]$ and leaves $n-k$ agents
unmatched; dummy chores are not included. The quantities $\val_w(F)$,
$\OPT(w)$, and $V(w,v)$ are as defined in
\Cref{sec:EFM:<n-chores:>0}.

\begin{lemma}[One alternating-path bound]\label{lem:expansion-exchange}
Suppose $\widehat w$ differs from $w$ only on edges incident to one
meta-chore. For any non-negative vector~$v$, any matching~$F$, and
distinct agents $i,j\in N(F)$,
\begin{equation}\label{eq:expansion-two-agents}
 \val_{\widehat w}(F)+v_i+v_j
 \leq V(\widehat w,v)+V(w,v)-\OPT(w).
\end{equation}
\end{lemma}
\begin{proof}
Fix a maximum-weight matching $F^0$ for~$w$.
Selecting any $a\in N(F^0)$ gives $\OPT(w)+v_a\leq V(w,v)$.
Thus, if either $i$ or $j$ is matched in~$F^0$, selecting the other
agent in $V(\widehat w,v)$ proves the bound.

Otherwise, $i$ and $j$ are the entering endpoints of two distinct
alternating paths in $F\mathbin{\triangle}F^0$: every meta-chore is
matched in both matchings, so each path introduces just one new agent.
At most one of these paths contains the meta-chore whose weights
changed. Choose the other path, say the one entering~$j$, and let
$g$ be the increase in $w$-weight when it is applied to~$F^0$.
The resulting matching contains~$j$, so
$\OPT(w)+g+v_j\leq V(w,v)$.
Reverting this path in~$F$ leaves a matching that still contains~$i$
and has $\widehat w$-weight $\val_{\widehat w}(F)-g$.
Hence $\val_{\widehat w}(F)-g+v_i\leq V(\widehat w,v)$.
Adding these inequalities proves \eqref{eq:expansion-two-agents}.
\end{proof}

For each $q\in[k]$, let $w^{-q}$ be obtained from~$w$ by replacing
$w(i,B_q)$ by $u_i(c_q)$ for every $i\in N$.
This resets just the meta-chore~$B_q$ to its original chore~$c_q$.
Its attachments are not returned to~$R$, so the vector~$v$ stays fixed.
We maintain, separately for every attachment set $S_q\ne\emptyset$,
the invariant
\begin{equation}\label{eq:reset-invariant}
\quad V(w^{-q},v)\leq \OPT(w).\quad
\end{equation}
Thus, after resetting any one previously expanded meta-chore, even
the best next expansion cannot exceed the current maximum welfare.

\begin{proof}[Proof of \Cref{lem:expansion-attachment-safety}]
We first prove \eqref{eq:reset-invariant} by induction on the number
of expansions. It holds initially because every attachment set is
empty. Suppose $D=R_r^+$ is attached to $c_p$, and use primes for
the new state. The choice in \Cref{alg:chore-expansion}
gives
\begin{equation}\label{eq:expansion-welfare-gain}
 \OPT(w')=V(w,v).
\end{equation}
Indeed, the selected matching attains the right-hand side after the
expansion. Conversely, in any matching of the new meta-chores, if
agent~$i$ receives $B_p$, the weight is the old matching weight plus
$u_i(D)\leq v_i$, and hence is at most $V(w,v)$.
Also, $v'_i\leq v_i$ and
\begin{equation}\label{eq:expansion-potential-decrease}
 u_i(D)\leq v_i-v'_i\leq v_i
 \qquad(i\in N),
\end{equation}
since $v_i-v'_i=\sum_{h\in D}\max\{u_i(h),0\}$.

Consider a previously expanded meta-chore~$B_q$ with $q\ne p$.
Take any matching~$F$, let $i$ receive~$B_p$, and select any
$j\in N(F)$. Its new score after resetting~$B_q$ is
\[
 \val_{(w')^{-q}}(F)+v'_j
 =\val_{w^{-q}}(F)+u_i(D)+v'_j.
\]
If $i=j$, \eqref{eq:expansion-potential-decrease} bounds this by
$\val_{w^{-q}}(F)+v_i\leq V(w^{-q},v)\leq \OPT(w)\leq \OPT(w')$.
If $i\ne j$, \Cref{lem:expansion-exchange}, with $\widehat w=w^{-q}$,
bounds it by
\[
 \val_{w^{-q}}(F)+v_i+v_j
 \leq V(w^{-q},v)+V(w,v)-\OPT(w)
 \leq V(w,v)=\OPT(w').
\]
Maximizing over~$F$ and~$j$ proves the new invariant for $q\ne p$.

For $q=p$, resetting removes the new attachments as well. If $B_p$
was already expanded, $(w')^{-p}=w^{-p}$ and $v'\leq v$ give
\[
 V((w')^{-p},v')\leq V(w^{-p},v)\leq \OPT(w)\leq \OPT(w').
\]
If this is its first expansion, $(w')^{-p}=w$, so
$V((w')^{-p},v')\leq V(w,v)=\OPT(w')$.
This completes the induction.

Finally, take any maximum-weight matching~$F$ and let~$i$ receive
$B_q=c_q\circ S_q$. If $S_q=\emptyset$, the claim is immediate.
Otherwise, resetting~$B_q$ reduces the weight of this same matching
by exactly $u_i(S_q)$. Selecting~$i$ in the reset problem and applying
\eqref{eq:reset-invariant} gives
\[
 \OPT(w)-u_i(S_q)+v_i
 =\val_{w^{-q}}(F)+v_i
 \leq V(w^{-q},v)\leq \OPT(w).
\]
Therefore $u_i(S_q)\geq v_i\geq0$. In particular, this holds for
every chore recipient in the final matching~$F^*$, proving
Observation~1 in the proof of \Cref{prop:EF1+EFable:<n-chores:>0}.
\end{proof}
\end{document}